%% file: main.tex
\documentclass[11pt,letterpaper]{article}

\usepackage[T1]{fontenc}
\usepackage{lmodern}
\usepackage[margin=1in]{geometry}
\usepackage{amsmath,amssymb,amsthm,mathtools}
\usepackage{booktabs,array}
\usepackage{graphicx}
\usepackage{flafter}
\usepackage{placeins}
\usepackage{float}
\floatstyle{ruled}
\newfloat{algofloat}{tbp}{loa}
\floatname{algofloat}{Algorithm}
\newenvironment{algorithm}[1][tb]{%
  \begin{algofloat}[#1]\centering
  \begin{minipage}{0.94\linewidth}\small
}{%
  \end{minipage}
  \end{algofloat}
}
\newcommand{\KwIn}[1]{\par\noindent\textbf{Input:} #1\par}
\newcommand{\KwOut}[1]{\par\noindent\textbf{Output:} #1\par}
\newcommand{\While}[2]{%
  \par\noindent\textbf{while }#1\textbf{ do}\par
  \begingroup\advance\leftskip by 1.5em #2\par\endgroup
}
\newcommand{\If}[2]{%
  \par\noindent\textbf{if }#1\textbf{ then}\par
  \begingroup\advance\leftskip by 1.5em #2\par\endgroup
}
\newcommand{\Return}[1]{\par\noindent\textbf{return }#1\par}
\usepackage{xcolor}
\usepackage{tikz}
\usetikzlibrary{arrows.meta,decorations.pathreplacing,calc}
\usepackage{microtype}
\usepackage{hyperref}
\usepackage[nameinlink,noabbrev]{cleveref}
\crefname{algofloat}{Algorithm}{Algorithms}
\Crefname{algofloat}{Algorithm}{Algorithms}

\hypersetup{
  colorlinks=true,
  linkcolor=blue!60!black,
  citecolor=blue!60!black,
  urlcolor=blue!60!black,
  breaklinks=true,
  bookmarksnumbered=true,
  pdftitle={Homological Trimming and Regularity of Filtrations via Local
            Obstruction Modules},
  pdfauthor={Siddharth Pritam},
  pdfsubject={Computational topology, persistent homology},
  pdfkeywords={persistent homology, simplicial collapse, flag complexes,
               local homology, preprocessing}
}

\newtheorem{theorem}{Theorem}[section]
\newtheorem{lemma}[theorem]{Lemma}
\newtheorem{proposition}[theorem]{Proposition}
\newtheorem{corollary}[theorem]{Corollary}

\theoremstyle{definition}
\newtheorem{definition}[theorem]{Definition}
\newtheorem{example}[theorem]{Example}

\newcommand{\redH}{\widetilde{\mathrm H}}
\newcommand{\lk}{\operatorname{Lk}}
\newcommand{\cl}{\operatorname{Cl}}
\newcommand{\Dgm}{\operatorname{Dgm}}
\newcommand{\RPD}{\operatorname{RPD}}
\newcommand{\RSD}{\operatorname{RSD}}
\newcommand{\Oc}{\mathcal O}
\newcommand{\Bc}{\mathcal B}
\newcommand{\eps}{\varepsilon}

\title{Homological Trimming and Regularity of Filtrations\\
       via Local Obstruction Modules}
\author{Siddharth Pritam\thanks{Chennai Mathematical Institute, Chennai, India.
  \texttt{spritam@cmi.ac.in}}}
\date{\today}

\begin{document}

\maketitle
\input{sections/abstract}

\input{sections/introduction}
\input{sections/running_figure}
\input{sections/preliminaries}
\input{sections/local_obstruction_modules}
\input{sections/deletion_certificates}
\input{sections/trimming_algorithm}
\input{sections/batch_deletion}
\input{sections/regularity_descriptors}
\input{sections/experiments}
\input{sections/conclusion}
\input{sections/backmatter}

\input{sections/references}
\appendix
\input{sections/appendix_missing_proofs}
\input{sections/appendix_endpoint_blockers}
\input{sections/appendix_rsd_experiment}

\end{document}

%% file: sections/abstract.tex
\begin{abstract}
Existing link-based combinatorial preprocessing methods speed up the computation
of persistent homology by removing a vertex or edge only when its link remains
a cone. We replace this condition with a quantitative homological certificate.
The reduced homology of the filtered link of a generator---a vertex or
edge---defines a local obstruction module whose future part describes the
effect of deleting that generator. Its barcode certifies either exact deletion
or an explicit bound on the bottleneck error, and a conflict colouring extends
this guarantee to families of generators. Our implementation,
\textsc{HomTrim}, removes a further $13\%–41\%$ of the input edges beyond domination-only preprocessing and reduces backend persistence time by
$1.55\times$--$4.61\times$ on weighted flag filtrations. The same module also
yields regularity diagrams that measure how far a generator can move before
becoming visible to homology, together with a multiscale stability result.

\end{abstract}

%% file: sections/introduction.tex
\section{Introduction}
\label{sec:intro}
Persistent homology (PH) is a central tool in topological data analysis, but
its computation can be expensive. Standard matrix-reduction algorithms have
worst-case complexity $O(n^3)$ in the number of simplices, although they are
often much faster in practice
\cite{ZomorodianCarlsson2005,EdelsbrunnerHarer2010}. The practical difficulty is
that a filtered complex may contain enormous numbers of simplices, many of
which can be removed without changing its persistence barcode. Preprocessing
therefore seeks to simplify the complex before persistence is computed, for
instance by discrete Morse theory \cite{MischaikowNanda2013}, by reductions of
the filtered complex \cite{DlotkoWagner2014}, or, for clique complexes, by
passing to a smaller simplicial set with the same homology
\cite{Zomorodian2010}. Strong collapses and edge collapses
\cite{BoissonnatPritamPareek2018,BoissonnatPritam2019,BoissonnatPritam2020,GlissePritam2022},
building
on Barmak and Minian's theory of strong homotopy types
\cite{BarmakMinian2012}, certify such simplifications through local link
conditions: a vertex or edge can be removed when it remains
\emph{dominated}---equivalently, when its link remains a simplicial cone---at
every relevant future parameter value.

Conicality is stronger than what is required homologically. A link may be
acyclic without being a cone, so a generator can be homologically harmless
even when no domination certificate is available. Moreover, the existing
link-based certificate is all-or-nothing: it does not record how long an
obstruction persists, and therefore gives no direct way to exchange a
controlled amount of error for further reduction. We address both points with
a local persistence module whose barcode gives a quantitative deletion
certificate and also measures the freedom around a generator's filtration
value.

\paragraph{The local obstruction module.}
Let $\mathcal F$ be the filtration, write $K_t$ for the complex at parameter
$t$, and let $K$ be the final complex. A \emph{generator} is a vertex in a
lower-star filtration or an edge in an edge-filtered flag complex; its arrival
time is denoted by $x_\gamma$. We record how the link of $\gamma$ changes while
the filtration values of all other generators remain fixed.

For a lower-star filtration induced by a vertex function
$f:K^{(0)}\to\mathbb R$, a generator is a vertex $v$ and $x_v=f(v)$. For an
edge-filtered flag complex, let $G$ be the final graph and
$w:E(G)\to\mathbb R$ its edge-weight function; a generator is an edge $e$ and
$x_e=w(e)$. Writing $\lk_K(\gamma)$ for the simplicial link in $K$, define
$$
\begin{aligned}
A_v(t)
&=\{\tau\in\lk_K(v):f(u)\leq t
      \text{ for every vertex }u\in\tau\},\\
A_e(t)
&=\{\tau\in\lk_K(e):w(e')\leq t
      \text{ for every edge }e'\subseteq\tau\cup e,\ e'\neq e\}.
\end{aligned}
$$
The first line retains the part of the link whose vertices have appeared by
time $t$. The second retains a simplex when all edges needed to form it with
$e$, apart from $e$ itself, have appeared. This uses only the common
neighbourhood of the endpoints of $e$, not the full ambient flag complex.

In either setting, $A_\gamma(s)\subseteq A_\gamma(t)$ for $s\leq t$. Taking
reduced homology over a fixed field therefore gives, in each degree $q$, the
persistence module
$$
\Oc_{\gamma,q}(t)=\redH_q(A_\gamma(t)).
$$
Together these form the \emph{local obstruction module} of $\gamma$. We call
$\gamma$ \emph{homologically regular at $t$} when
$\Oc_{\gamma,q}(t)=0$ for every $q$, or equivalently, when $A_\gamma(t)$ is
acyclic over the chosen field.

\paragraph{The deletion certificate.}
Regularity at the arrival time alone does not justify deletion: a generator
may be harmless when it appears and become part of a nontrivial cycle later.
We therefore retain the obstruction module from $x_\gamma$ onward. In degree
$q$, the \emph{future obstruction module} is
$$
\Oc_{\gamma,q}^{\geq x_\gamma}(t)
=
\begin{cases}
0, & t<x_\gamma,\\
\Oc_{\gamma,q}(t), & t\geq x_\gamma.
\end{cases}
$$
Let $\mathcal F\setminus\gamma$ be the filtration obtained by deleting
$\gamma$ and all its cofaces, so its complex at $t$ is
$K_t\setminus\gamma=\{\sigma\in K_t:\gamma\not\subseteq\sigma\}$. The relative
groups $H_k(K_t,K_t\setminus\gamma)$ form a persistence module
$H_k(\mathcal F,\mathcal F\setminus\gamma)$ that measures the effect of this
deletion throughout the filtration.

The local and global modules are, up to a degree shift, the same object.

\begin{theorem}[\Cref{thm:future-relative}]
For every $k$,
$$
H_k(K_t,K_t\setminus\gamma)
\cong
\begin{cases}
0, & t<x_\gamma,\\
\redH_{k-\dim\gamma-1}(A_\gamma(t)), & t\geq x_\gamma.
\end{cases}
$$
These isomorphisms commute with the filtration maps and hence give
$$
H_k(\mathcal F,\mathcal F\setminus\gamma)
\cong
\Oc_{\gamma,k-\dim\gamma-1}^{\geq x_\gamma}
$$
as persistence modules.
\end{theorem}

The shift comes from joining a link class with $\gamma$. More importantly, the
theorem turns a computation in one filtered link into a global certificate.
Writing $\Dgm_k(\mathcal F)$ for the degree-$k$ persistence diagram and $d_B$
for bottleneck distance, the long exact sequence of the pair and the induced
matching theorem of Bauer and Lesnick \cite{BauerLesnick2015} give:

\begin{corollary}[\Cref{cor:trimming}]
If every interval in the barcode of
$\Oc_{\gamma,q}^{\geq x_\gamma}$ has length at most $\eps$, for every $q$, then
for every $k$,
$$
d_B\bigl(\Dgm_k(\mathcal F),
          \Dgm_k(\mathcal F\setminus\gamma)\bigr)
\leq\eps.
$$
In particular, if the future obstruction module vanishes in every degree,
then deleting $\gamma$ preserves persistent homology.
\end{corollary}

Thus the longest future obstruction bar bounds the cost of deleting $\gamma$
and its cofaces. This is a homological certificate: the link need not be a
cone, and the deletion need not be a collapse. \Cref{fig:running} illustrates
the argument for one edge.

\paragraph{The trimming algorithm.}
The certificate is local but not static: one deletion may change nearby
links. The algorithm therefore tests a candidate in the current complex,
deletes it when its future obstruction meets the chosen tolerance, and
re-queues only the candidates that share a simplex with the deleted one, which
are the only ones whose links can change. Each deletion is thus certified in
the complex in which it is performed.

Let $\rho(\gamma_i)$ be the longest future obstruction bar when
$\gamma_i$ is deleted. If a run deletes
$\gamma_1,\ldots,\gamma_m$ and produces a filtration $\mathcal F'$, then
$$
\eta=\sum_{i=1}^m\rho(\gamma_i)
\qquad\text{and}\qquad
d_B\bigl(\Dgm_k(\mathcal F),\Dgm_k(\mathcal F')\bigr)\leq\eta
$$
in every certified degree $k$ (\Cref{thm:sequential-trimming}). Zero cost gives
exact preservation. The bound holds for any processing order, although the
reduction may depend on that order, and each test depends on a local lower-link
filtration rather than the ambient complex.

\paragraph{Deleting many generators at once.}
The additive bound charges every deletion separately. Deletions interact only
through simplices that contain two of them, and when there are no such
simplices the bound improves. If no simplex contains two members of a family
$S$, the relative module of deleting $S$ splits as a direct sum of the future
obstruction modules of its members (\Cref{thm:batch}), so
$$
d_B\bigl(\Dgm_k(\mathcal F),\Dgm_k(\mathcal F\setminus S)\bigr)
\leq\max_{\gamma\in S}\rho(\gamma).
$$
A run can proceed in rounds of such families. One pass over a greedy
colouring of the candidates of cost at most $\eps$ then certifies an error of
at most $(\Delta_c+1)\eps$, where $\Delta_c$ is the largest number of these
candidates that share a simplex with any one of them. The
bound does not depend on how many generators are removed. For deletions that
do interact, we know no bound better than the sum.

\paragraph{Experiments.}
We implement the sequential algorithm for weighted flag filtrations, using
Ripser \cite{Bauer2021} as the backend. Averaged over our test families,
\textsc{HomTrim} removes a further $13\%$--$41\%$ of the input edges beyond
domination-only preprocessing and makes the backend $1.55\times$--$4.61\times$
faster. The
local homology has its own cost, however, and in our current implementation it
is not recovered against domination-only preprocessing. Against direct
computation, the complete pipeline becomes competitive as the input grows
denser: on our sparse-matrix family it is $1.24\times$ faster at average
degree $48$, and on an instance of degree $64$ it finishes while the untrimmed
backend times out. We report where the preprocessing pays for itself and where it does not.

\paragraph{Regularity slack.}
The obstruction module also has a descriptive reading. Suppose that $\gamma$
is regular at $x_\gamma$, and consider the largest interval containing
$x_\gamma$ on which $\Oc_{\gamma,q}(t)=0$ in every degree. The distances from
$x_\gamma$ to the two ends of this interval are its backward and forward
slacks: within this window the value of $\gamma$ can move, with all other
values fixed, without changing persistent homology. Recording this pair for
every regular generator gives the \emph{regularity slack diagram}. It
measures filtration events that can move without becoming visible to
homology, information that is not recorded by the ordinary persistence
diagram.

Exact slack is not Lipschitz stable: an arbitrarily short obstruction bar can
split a long zero interval. The obstruction barcode itself is stable, however.
If bars below a chosen scale are ignored before slack is measured, the
resulting obstruction sets satisfy a scale-shifted stability bound. We develop
the exact and multiscale versions in \Cref{sec:diagrams}. The accompanying
experiment is deliberately small: synthetic examples distinguish local
mechanisms, but we do not claim application-level validation.

\paragraph{Relation to previous work.}
The ingredients of the arguments of our results are standard. The lower link comes from
piecewise-linear Morse theory, where a vertex is classified by the topology of
the link below it \cite{Banchoff1970,EdelsbrunnerHarer2010}; our regularity
condition is that classification, read at every parameter rather than at one.
Local homology has been used to infer structure from samples
\cite{BendichWangMukherjee2012,SkrabaWang2014}. The two results we lean on are
classical as well: the attachment computation behind \Cref{thm:future-relative},
and the induced matching theorem \cite{BauerLesnick2015}, which is what turns a
short relative module into a bottleneck bound. What is new is the use we make
of them. We follow the homology of a single link through the filtration and
identify the resulting local module with the relative persistence of a
deletion. Nothing in that computation depends on the dimension of $\gamma$, so
it covers generators of every dimension, although we develop algorithms only
for vertices and edges. Applied to a whole family at once, it also gives the
splitting of \Cref{thm:batch}.

Collapse-based preprocessing is the closest of the methods mentioned above: it
also certifies one generator at a time, and \Cref{cor:trimming} contains its
domination test, as the case in which acyclicity is witnessed by an apex. The
others rebuild the filtration instead---as a Morse complex
\cite{MischaikowNanda2013}, a reduced chain complex \cite{DlotkoWagner2014}, or
a smaller simplicial set \cite{Zomorodian2010}---while approximation schemes
such as metric sparsification \cite{Sheehy2013} and controlled changes of
filtration values \cite{NigmetovMorozov2024} alter the input globally and bound
the effect afterwards. Here the tolerance is spent one generator at a time, and
each deletion is charged its cost, the length of its longest future obstruction
bar (\Cref{thm:sequential-trimming}).

\paragraph{Scope and organization.}
We work over a fixed field; the certificates do not imply homotopy equivalence,
and the greedy algorithm need not be optimal. General simplex-wise filtrations
satisfy the theorem, but we develop the algorithmic details only for lower-star
and edge-filtered flag filtrations. \Cref{sec:preliminaries} gives the
background, and \Cref{sec:obstruction} defines the local obstruction module.
\Cref{sec:certificate} proves the deletion theorem, \Cref{sec:algorithm} turns
it into a sequential algorithm, and \Cref{sec:batch} treats families of
deletions that can be made together. \Cref{sec:diagrams} returns to the
descriptive side of the module, \Cref{sec:experiments} reports the
experiments, and \Cref{sec:conclusion} collects open questions. Deferred
proofs, the homological degrees at regularity endpoints, and the controlled
descriptor experiment appear in the appendices.

%% file: sections/running_figure.tex
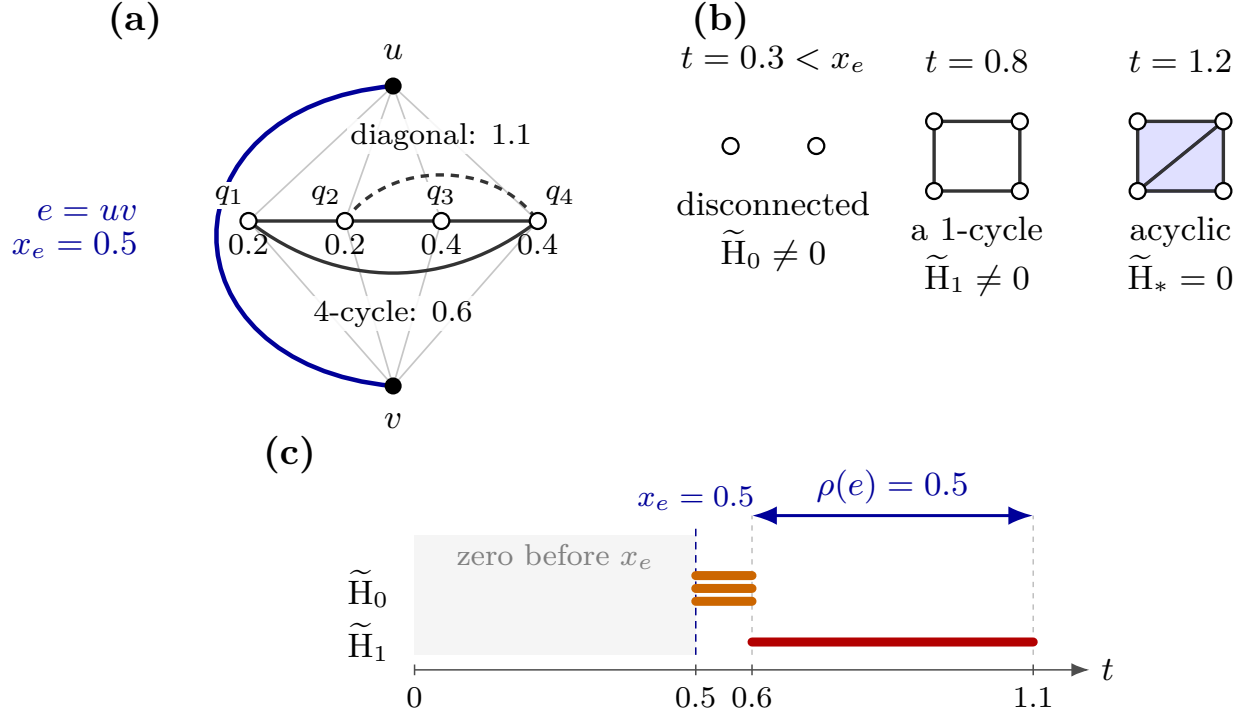
\begin{figure}[ht]
\centering
\resizebox{\linewidth}{!}{%
\begin{tikzpicture}[
  x=1cm,y=1cm,>=Latex,
  vtx/.style={circle,fill=black,inner sep=1.6pt},
  ovtx/.style={circle,draw=black,fill=white,line width=0.6pt,
               inner sep=1.5pt},
  panel/.style={font=\small\bfseries},
  lbl/.style={font=\footnotesize},
  slbl/.style={font=\scriptsize},
  fan/.style={gray!45,line width=0.45pt},
  lk/.style={black!80,line width=0.9pt},
  tri/.style={fill=blue!12,draw=blue!55!black,line width=0.7pt},
  bar/.style={line width=2.4pt,line cap=round}
]

\begin{scope}[shift={(0,0)}]
  \node[panel,anchor=south west] at (-1.45,2.48) {(a)};

  \coordinate (u) at (1.35,2.18);
  \coordinate (v) at (1.35,-0.62);
  \coordinate (q1) at (0.00,0.92);
  \coordinate (q2) at (0.90,0.92);
  \coordinate (q3) at (1.80,0.92);
  \coordinate (q4) at (2.70,0.92);

  \foreach \q in {q1,q2,q3,q4}{
    \draw[fan] (u)--(\q);
    \draw[fan] (v)--(\q);
  }

  \draw[blue!60!black,line width=1.2pt]
    (u) .. controls (-0.85,1.98) and (-0.85,-0.42) .. (v);
  \node[lbl,blue!60!black,anchor=east,align=right]
    at (-0.90,0.82) {$e=uv$\\[-2pt]$x_e=0.5$};

  \draw[lk] (q1)--(q2)--(q3)--(q4);
  \draw[lk] (q1) to[bend right=38] (q4);
  \draw[lk,dash pattern=on 2.2pt off 1.8pt]
    (q2) to[bend left=55] (q4);

  \node[vtx] at (u) {};
  \node[vtx] at (v) {};
  \node[lbl,anchor=south] at (1.35,2.30) {$u$};
  \node[lbl,anchor=north] at (1.35,-0.74) {$v$};
  \foreach \q/\name/\time/\dx in {
      q1/q_1/0.2/-5pt,q2/q_2/0.2/-5pt,q3/q_3/0.4/0pt,q4/q_4/0.4/6pt}{
    \node[ovtx] at (\q) {};
    \node[slbl,anchor=south,inner sep=1pt,fill=white]
      at ([xshift=\dx,yshift=3pt]\q) {$\name$};
    \node[slbl,anchor=north,inner sep=2.5pt]
      at ([yshift=-1pt]\q) {$\time$};
  }

  \node[slbl,anchor=south,fill=white,fill opacity=0.85,text opacity=1,
        inner sep=1pt] at (1.80,1.54) {diagonal: $1.1$};
  \node[slbl,anchor=north,fill=white,fill opacity=0.85,text opacity=1,
        inner sep=1pt] at (1.35,0.22) {$4$-cycle: $0.6$};
\end{scope}

\begin{scope}[shift={(4.45,0)}]
  \node[panel,anchor=south west] at (-0.45,2.48) {(b)};

  \begin{scope}[shift={(0,0)}]
    \node[lbl,anchor=south] at (0.45,2.18) {$t=0.3<x_e$};
    \node[ovtx] at (0.05,1.62) {};
    \node[ovtx] at (0.85,1.62) {};
    \node[lbl,anchor=north,align=center,inner sep=2pt] at (0.45,1.27)
      {disconnected\\[2pt]$\redH_0\neq0$};
  \end{scope}

  \begin{scope}[shift={(1.90,0)}]
    \node[lbl,anchor=south] at (0.45,2.18) {$t=0.8$};
    \coordinate (b1) at (0.05,1.85);
    \coordinate (b2) at (0.85,1.85);
    \coordinate (b3) at (0.85,1.20);
    \coordinate (b4) at (0.05,1.20);
    \draw[lk] (b1)--(b2)--(b3)--(b4)--(b1);
    \foreach \p in {b1,b2,b3,b4}{\node[ovtx] at (\p) {};}
    \node[lbl,anchor=north,align=center,inner sep=2pt] at (0.45,1.04)
      {a $1$-cycle\\[2pt]$\redH_1\neq0$};
  \end{scope}

  \begin{scope}[shift={(3.80,0)}]
    \node[lbl,anchor=south] at (0.45,2.18) {$t=1.2$};
    \coordinate (c1) at (0.05,1.85);
    \coordinate (c2) at (0.85,1.85);
    \coordinate (c3) at (0.85,1.20);
    \coordinate (c4) at (0.05,1.20);
    \filldraw[tri] (c1)--(c2)--(c4)--cycle;
    \filldraw[tri] (c2)--(c3)--(c4)--cycle;
    \draw[lk] (c1)--(c2)--(c3)--(c4)--(c1);
    \draw[lk] (c2)--(c4);
    \foreach \p in {c1,c2,c3,c4}{\node[ovtx] at (\p) {};}
    \node[lbl,anchor=north,align=center,inner sep=2pt] at (0.45,1.04)
      {acyclic\\[2pt]$\redH_*=0$};
  \end{scope}
\end{scope}

\begin{scope}[shift={(1.55,-2.55)}]
  \node[panel,anchor=south west] at (-1.55,0.98) {(c)};
  \def\sx{5.25}

  \fill[gray!8] (0,-0.58) rectangle ({0.5*\sx},0.54);
  \node[slbl,text=black!50] at ({0.25*\sx},0.34)
    {zero before $x_e$};

  \draw[blue!55!black,dash pattern=on 2pt off 1.5pt]
    ({0.5*\sx},-0.58) -- ({0.5*\sx},0.62);
  \node[slbl,blue!60!black,anchor=south]
    at ({0.5*\sx},0.64) {$x_e=0.5$};

  \draw[black!25,dash pattern=on 1.4pt off 1.6pt]
    ({0.6*\sx},-0.58) -- ({0.6*\sx},0.72);
  \draw[black!25,dash pattern=on 1.4pt off 1.6pt]
    ({1.1*\sx},-0.58) -- ({1.1*\sx},0.72);

  \draw[<->,blue!60!black,line width=0.8pt]
    ({0.6*\sx},0.72) -- ({1.1*\sx},0.72);
  \node[lbl,blue!60!black,anchor=south,inner sep=2pt]
    at ({0.85*\sx},0.76) {$\rho(e)=0.5$};

  \node[lbl,anchor=east] at (-0.10,0.04) {$\redH_0$};
  \node[lbl,anchor=east] at (-0.10,-0.46) {$\redH_1$};

  \draw[bar,orange!80!black] ({0.5*\sx},0.16) -- ({0.6*\sx},0.16);
  \draw[bar,orange!80!black] ({0.5*\sx},0.04) -- ({0.6*\sx},0.04);
  \draw[bar,orange!80!black] ({0.5*\sx},-0.08) -- ({0.6*\sx},-0.08);
  \draw[bar,red!70!black] ({0.6*\sx},-0.46) -- ({1.1*\sx},-0.46);

  \draw[->,black!70] (0,-0.72) -- ({1.20*\sx},-0.72);
  \node[lbl,anchor=west,inner sep=3pt] at ({1.20*\sx},-0.72) {$t$};
  \foreach \t/\lab in {0/0,0.5/0.5,0.6/0.6,1.1/1.1}{
    \draw[black!70] ({\t*\sx},-0.67) -- ({\t*\sx},-0.77);
    \node[slbl,anchor=north,inner sep=2.5pt]
      at ({\t*\sx},-0.79) {$\lab$};
  }
\end{scope}

\end{tikzpicture}%
}
\caption{A local deletion certificate for an edge $e$ arriving at $x_e=0.5$.
(a) Both gray edges $uq_i$ and $vq_i$ appear by the value shown below $q_i$,
so $q_i$ may enter $A_e(t)$ before $e$ appears. The four cycle edges enter at
$0.6$, and the dashed diagonal at $1.1$. (b) The link is first disconnected,
then carries a $1$-cycle, and finally becomes acyclic. (c) The future
obstruction modules $\Oc_{e,q}^{\geq x_e}$ discard everything before $x_e$.
The link fails to be a cone at some future values, so domination cannot
remove $e$. For $p\geq2$, so that local degree one is tested, the longest
remaining bar has length $\rho(e)=0.5$, and \Cref{cor:trimming} certifies
deletion of $e$ and its cofaces with bottleneck error at most $0.5$. For
$p=1$ only the degree-zero bars count, and $\rho(e)=0.1$.}
\label{fig:running}
\end{figure}

%% file: sections/preliminaries.tex
\section{Preliminaries}\label{sec:preliminaries}
This section fixes the notation and recalls the homological facts used later.
Standard references include
\cite{EdelsbrunnerHarer2010,ZomorodianCarlsson2005}.

\paragraph{Simplicial complexes.}
An (abstract) simplicial complex on a finite vertex set $X$ is a collection
$K\subseteq 2^X$ such that every subset of a simplex in $K$ also belongs to
$K$. Its elements are called \emph{simplices}. A simplex with $k+1$ vertices
has dimension $k$, and it is \emph{maximal} if it is not properly contained in
another simplex. We write $K^{(k)}$ for the set of $k$-simplices. A
subcollection $L\subseteq K$ that is itself a simplicial complex is a
\emph{subcomplex} of $K$.

A \emph{simplicial map} $\varphi:K\to L$ is a map on vertices for which
$\varphi(\sigma)$ is a simplex of $L$ whenever $\sigma$ is a simplex of $K$.
For a graph $G=(V,E)$, its \emph{clique complex}, or \emph{flag complex},
$\cl(G)$ contains one simplex for every clique of $G$. It is therefore
completely determined by its one-skeleton.

\paragraph{Joins, stars, and links.}
Let $K$ and $L$ have disjoint vertex sets. Their \emph{join} is
$$
K*L=\{\sigma\cup\tau:\sigma\in K,\ \tau\in L\}.
$$
If $v$ is a new vertex, the join $v*K$ is the \emph{cone} over $K$ with apex
$v$. For a simplex $\sigma\in K$, its star and link are
$$
\operatorname{st}_K(\sigma)
=\{\tau\in K:\sigma\cup\tau\in K\},
\qquad
\lk_K(\sigma)
=\{\tau\in\operatorname{st}_K(\sigma):
  \tau\cap\sigma=\varnothing\}.
$$
When $\sigma$ is viewed together with all its faces, these definitions give
$\operatorname{st}_K(\sigma)=\sigma*\lk_K(\sigma)$.

\paragraph{Filtrations.}
A \emph{filtration function} $h:K\to\mathbb R$ assigns a value to every
simplex so that $h(\sigma)\leq h(\tau)$ whenever $\sigma\subseteq\tau$. The
complex present at parameter $t$ is
$$
K_t=\{\sigma\in K:h(\sigma)\leq t\}.
$$
These complexes increase with $t$. We denote the resulting filtration by
$\mathcal F$ and use $K_t$ for its level at $t$.

A vertex function $f:K^{(0)}\to\mathbb R$ induces the \emph{lower-star
filtration}
$$
K_t=\{\sigma\in K:f(u)\leq t
      \text{ for every vertex }u\in\sigma\}.
$$
Thus a simplex $\sigma$ enters at $\max_{u\in\sigma}f(u)$. Similarly, let
$G=(V,E)$ be a finite graph with an edge function $w:E\to\mathbb R$. If
$G_t$ contains all vertices and the edges $e$ with $w(e)\leq t$, then
$K_t=\cl(G_t)$ defines the corresponding \emph{edge-filtered flag complex}.
Each positive-dimensional simplex enters at the largest value of one of its
edges. We assume throughout that all vertices appear before the edges.

\paragraph{Homology.}
All homology is taken over a fixed field $\mathbb F$. Let $C_q(K)$ be the
vector space generated by the $q$-simplices of $K$, with boundary map
$\partial_q:C_q(K)\to C_{q-1}(K)$. The cycles and boundaries are
$Z_q(K)=\ker\partial_q$ and $B_q(K)=\operatorname{im}\partial_{q+1}$, and
$$
H_q(K)=Z_q(K)/B_q(K).
$$
If $L\subseteq K$ is a subcomplex, the quotient
$C_q(K,L)=C_q(K)/C_q(L)$ inherits a boundary map. Its homology $H_q(K,L)$ is
the \emph{relative homology} of the pair $(K,L)$; it treats every chain in
$L$ as zero.

Reduced homology, denoted by $\redH_q(K)$, is obtained from the augmented
chain complex. It agrees with ordinary homology in positive degrees, while
$\dim\redH_0(K)$ is one less than the number of connected components when
$K$ is nonempty. We use the convention
$\redH_{-1}(\{\varnothing\})\cong\mathbb F$.

\paragraph{Persistent homology.}
Let $\mathcal F$ be a filtration with complex $K_t$ at parameter $t$. For
$s\leq t$, the inclusion $K_s\hookrightarrow K_t$ induces a map
$H_q(K_s)\to H_q(K_t)$. The groups and these \emph{structure maps} form the
degree-$q$ persistence module, which we denote by $H_q(\mathcal F)$
\cite{EdelsbrunnerLetscherZomorodian2002,ZomorodianCarlsson2005}. Replacing
ordinary homology by reduced homology gives $\redH_q(\mathcal F)$.

Our filtrations are finite, so each persistence module decomposes uniquely
into interval modules \cite{CrawleyBoevey2015}. The resulting multiset of
intervals is its \emph{barcode}; an interval $[b,d)$ records a class born at
$b$ and dying at $d$. Representing this interval by $(b,d)$ gives the
\emph{persistence diagram}, denoted by $\Dgm_q(\mathcal F)$.

An isomorphism between persistence modules $V$ and $W$ consists of
isomorphisms $\varphi_t:V_t\to W_t$ that commute with every structure map.
Thus, if $v_{s,t}$ and $w_{s,t}$ are the respective maps from $s$ to $t$, then
$$
\varphi_t\circ v_{s,t}=w_{s,t}\circ\varphi_s
\qquad\text{for every }s\leq t.
$$
Isomorphic modules have the same barcode. Approximate agreement is measured
by the \emph{interleaving distance} $d_I$: informally, an $\eps$-interleaving
allows the two modules to map to one another after shifting the parameter by
$\eps$. Persistence diagrams are compared by the \emph{bottleneck distance}
$d_B$, which minimizes the largest $\ell_\infty$ cost of a matching, allowing
points to be matched to the diagonal. For the finite modules considered here,
the isometry theorem gives
$$
d_I(V,W)=d_B\bigl(\Dgm(V),\Dgm(W)\bigr)
$$
\cite{ChazalEtAl2016,Lesnick2015}.

\paragraph{Exact sequences and relative persistence.}
A sequence of vector spaces and linear maps is \emph{exact} if, at each term,
the image of the incoming map is the kernel of the outgoing map. In
particular, a sequence $0\to U\to V\to W\to0$ is short exact when the first
map is injective, the second is surjective, and the image of the first is the
kernel of the second.

For a pair $L\subseteq K$, the relative chain complex fits into the short
exact sequence
$$
0\longrightarrow C_*(L)
\longrightarrow C_*(K)
\longrightarrow C_*(K,L)
\longrightarrow0.
$$
It induces the long exact sequence
$$
\cdots\longrightarrow H_q(L)
\longrightarrow H_q(K)
\longrightarrow H_q(K,L)
\xrightarrow{\delta_q}H_{q-1}(L)
\longrightarrow\cdots,
$$
where $\delta_q$ is the connecting homomorphism.

This construction is natural. A simplicial map of pairs
$\varphi:(K,L)\to(K',L')$, meaning that $\varphi(L)\subseteq L'$, induces
maps on all three homology groups, and these maps commute with every arrow in
the two long exact sequences. Write
$\varphi^L_q$ and $\varphi^{\mathrm{rel}}_q$ for the induced maps on
$H_q(L)$ and $H_q(K,L)$, respectively. If $\delta_q$ and $\delta'_q$ are the
connecting maps, then
$$
\varphi^L_{q-1}\circ\delta_q
=\delta'_q\circ\varphi^{\mathrm{rel}}_q.
$$
Both sides carry $H_q(K,L)$ to $H_{q-1}(L')$.

Now let $\mathcal L\subseteq\mathcal F$ be a filtered pair, with
$L_t\subseteq K_t$ at every parameter. The relative groups
$H_q(K_t,L_t)$ and the maps induced by the inclusions form the
\emph{relative persistence module}, denoted by
$$
R_q=H_q(\mathcal F,\mathcal L).
$$
Similarly, $H_q(\mathcal L)$ denotes the persistence module formed by the
groups $H_q(L_t)$.
At each $t$, the pair $(K_t,L_t)$ has the long exact sequence above.
Naturality says that these sequences commute with the structure maps from
$s$ to $t$. Consequently, the levelwise sequences assemble into a long exact
sequence of persistence modules:
$$
\cdots\longrightarrow H_q(\mathcal L)
\longrightarrow H_q(\mathcal F)
\longrightarrow R_q
\longrightarrow H_{q-1}(\mathcal L)
\longrightarrow\cdots.
$$
Here exactness means exactness at every parameter. Thus the two relative
modules adjacent to $H_q(\mathcal L)\to H_q(\mathcal F)$ control whether this
map is an isomorphism.

\paragraph{Consequences of relative persistence.}
We first record the exact case, a standard consequence of the long exact
sequence \cite[Section~IV.4]{EdelsbrunnerHarer2010}.

\begin{lemma}[Vanishing relative module]
\label{lem:zero-relative-module}
If $R_q=R_{q+1}=0$, then the inclusion induces an isomorphism
$H_q(\mathcal L)\to H_q(\mathcal F)$. In particular, if $R_q=0$ for every
$q$, then it induces an isomorphism in every degree.
\end{lemma}

\begin{proof}
The relevant part of the long exact sequence is
$$
R_{q+1}\longrightarrow H_q(\mathcal L)
\longrightarrow H_q(\mathcal F)\longrightarrow R_q.
$$
The first vanishing gives injectivity and the second gives surjectivity.
\end{proof}

The following local form is the corresponding excision statement
\cite[Section~IV.3]{EdelsbrunnerHarer2010}.

\begin{corollary}[Local attachment principle]
\label{lem:local-attachment-principle}
Let $K=L\cup B$ and $C=L\cap B$. If $H_*(B,C)=0$, then
$L\hookrightarrow K$ induces an isomorphism on homology. More generally,
suppose $K_t=L_t\cup B_t$ at every parameter, where both $L_t$ and $B_t$
increase with $t$, and set $C_t=L_t\cap B_t$. If
$H_*(B_t,C_t)=0$ for every $t$, then the inclusions
$L_t\hookrightarrow K_t$ induce an isomorphism of persistence modules.
\end{corollary}

\begin{proof}
The quotient chain complexes $C_*(K,L)$ and $C_*(B,C)$ are naturally
isomorphic. Hence $H_*(K,L)=0$, and the first claim follows from the long
exact sequence. The same identification at every parameter commutes with the
structure maps, so the filtered claim follows from
\Cref{lem:zero-relative-module}.
\end{proof}

We also need a quantitative version. Following
\cite[Section~6]{BauerLesnick2015}, a persistence module $M$ is
\emph{$\eps$-trivial} if every structure map $M_s\to M_t$ with
$t-s\geq\eps$ is zero. For a finite interval-decomposable module, this is
equivalent to every barcode interval having length at most $\eps$.

\begin{lemma}[Approximate relative module]
\label{lem:approx-relative-module}
If $R_q$ and $R_{q+1}$ are $\eps$-trivial, then
$$
d_I\bigl(H_q(\mathcal L),H_q(\mathcal F)\bigr)\leq\eps,
\qquad
d_B\bigl(\Dgm_q(\mathcal L),\Dgm_q(\mathcal F)\bigr)\leq\eps.
$$
\end{lemma}

\begin{proof}
Exactness shows that the kernel of
$H_q(\mathcal L)\to H_q(\mathcal F)$ is a quotient of a submodule of
$R_{q+1}$, while its cokernel is a submodule of $R_q$. Both are therefore
$\eps$-trivial. The induced-matching theorem
\cite[Theorem~6.1]{BauerLesnick2015} gives the bottleneck bound, and the
isometry theorem gives the interleaving bound \cite{ChazalEtAl2016}.
\end{proof}

\paragraph{Cone and join attachments.}
The local pieces used later are cones, or more generally joins of a simplex
with its link. Their relative homology is determined by the reduced homology
of the link. We call a complex $A$ \emph{acyclic} if $\redH_*(A)=0$.

\begin{lemma}[Cone attachment]
\label{lem:cone-attachment}
Let $A$ be a finite simplicial complex and let $v*A$ be its cone. Then
$$
H_q(v*A,A)\cong\redH_{q-1}(A).
$$
Consequently, $A\hookrightarrow v*A$ induces an isomorphism on homology if
and only if $A$ is acyclic.
\end{lemma}

\begin{proof}
The cone $v*A$ is contractible, so the isomorphism follows from the long exact
sequence of $(v*A,A)$. The inclusion preserves homology precisely when all
relative groups vanish, equivalently when $A$ is acyclic.
\end{proof}

\begin{lemma}[Join attachment]
\label{lem:join-pair}
Let $\sigma$ be a simplex of dimension $d$ and let $A$ be a finite simplicial
complex on a disjoint vertex set. Writing $\partial\sigma$ for the boundary of
$\sigma$, we have
$$
H_q\bigl(\sigma*A,(\partial\sigma)*A\bigr)
\cong\redH_{q-d-1}(A).
$$
These relative groups vanish in every degree if and only if $A$ is acyclic.
\end{lemma}

\begin{proof}
The relative chains are generated by the simplices $\sigma\cup\tau$ with
$\tau\in A$, including $\tau=\varnothing$. Sending $\sigma\cup\tau$ to
$\tau$ identifies the relative chain complex with the augmented chain complex
of $A$, shifted by $d+1$. The boundary maps agree up to the global sign
$(-1)^{d+1}$, and the result follows.
\end{proof}

%% file: sections/local_obstruction_modules.tex
\section{Local obstruction modules}\label{sec:obstruction}

Fixing all filtration values except one generator leaves a simple local
question: over which parameter values does the star of that generator attach
without changing homology? The answer is encoded by a filtration of its link.
We define this filtration first for vertices and then for edges.

\subsection{Lower-star filtrations}

Let $f:K^{(0)}\to\mathbb R$ induce a lower-star filtration, denoted by
$\mathcal F_f$, and fix a vertex $v$ with value $x_v=f(v)$. At a parameter
$t$, define
$$
A_v^f(t)
=\{\tau\in\lk_K(v):f(u)\leq t
  \text{ for every vertex }u\in\tau\}.
$$
Thus $A_v^f(t)$ is the part of the link already present when the value of $v$
is placed at $t$. These complexes increase with $t$ and therefore form a
filtration of $\lk_K(v)$. At the original value of $v$ we recover its lower
link and lower star:
$$
\lk_f^-(v)=A_v^f(x_v),
\qquad
\operatorname{st}_f^-(v)=v*A_v^f(x_v).
$$

Applying reduced homology gives the local persistence module used throughout
the paper.

\begin{definition}[Vertex obstruction module]
\label{def:vertex-obstruction}
For each degree $q$, the \emph{vertex obstruction module} of $v$ is
$$
\Oc_{v,q}^f(t)=\redH_q(A_v^f(t)).
$$
We call $v$ \emph{homologically regular at $t$} if
$\Oc_{v,q}^f(t)=0$ for every $q$.
\end{definition}

The local filtration depends only on the values of vertices other than $v$.
The value $x_v$ is an anchor: it specifies where $v$ is placed inside a
filtration that is already determined by its neighbours. This makes it
possible to read from $\Oc_v^f$ how far the value of $v$ may move.

\begin{theorem}[Local-to-global criterion for lower-star filtrations]
\label{thm:ls-local-global}
Let $f'$ agree with $f$ away from $v$, and let $f'(v)=s$. If
$$
\Oc_{v,q}^f(t)=0
\qquad
\text{for every $q$ and every }
t\in[\min\{x_v,s\},\max\{x_v,s\}),
$$
then $H_q(\mathcal F_f)\cong H_q(\mathcal F_{f'})$ as persistence modules
for every $q$.
\end{theorem}

\begin{proof}
Suppose first that $x_v<s$. The two filtrations agree before $x_v$ and from
$s$ onward. For $x_v\leq t<s$, their levels satisfy
$K_t^{f'}\subseteq K_t^f$. Writing $A_t=A_v^f(t)$, the part present only in
$K_t^f$ is the cone $v*A_t$, attached along $A_t$:
$$
K_t^f=K_t^{f'}\cup(v*A_t),
\qquad
K_t^{f'}\cap(v*A_t)=A_t.
$$
By excision and \Cref{lem:cone-attachment},
$$
H_q(K_t^f,K_t^{f'})
\cong H_q(v*A_t,A_t)
\cong\redH_{q-1}(A_t)=0.
$$
The relative module therefore vanishes, and
\Cref{lem:zero-relative-module} gives the claimed persistence-module
isomorphism. The case $s<x_v$ is the same with the inclusions reversed.
\end{proof}

\subsection{Edge-filtered flag complexes}

Let $G=(V,E)$ be a finite graph, let $K=\cl(G)$, and let
$w:E\to\mathbb R$ define an edge-filtered flag complex $\mathcal F_w$. We
continue to assume that all vertices appear before the edges. For an edge
$e\in E$, write $x_e=w(e)$ and define
$$
A_e^w(t)
=\{\tau\in\lk_K(e):w(a)\leq t
  \text{ for every edge }a\subseteq e\cup\tau, a\neq e\}.
$$
A simplex belongs to $A_e^w(t)$ precisely when every edge needed to form it
together with $e$, apart from $e$ itself, has appeared by time $t$. The
complexes $A_e^w(t)$ form a filtration of $\lk_K(e)$, and
$$
\lk_w^-(e)=A_e^w(x_e),
\qquad
\operatorname{st}_w^-(e)=e*A_e^w(x_e).
$$
For $e=uv$, this filtration is supported on the common neighbours of $u$ and
$v$, as in edge collapse \cite{BoissonnatPritam2020}, and can therefore be
constructed without forming the full flag complex.

\begin{definition}[Edge obstruction module]
\label{def:edge-obstruction}
For each degree $q$, the \emph{edge obstruction module} of $e$ is
$$
\Oc_{e,q}^w(t)=\redH_q(A_e^w(t)).
$$
We call $e$ \emph{homologically regular at $t$} if
$\Oc_{e,q}^w(t)=0$ for every $q$.
\end{definition}

As before, $A_e^w$ is determined by the values of the other edges; $x_e$ only
anchors $e$ inside this local filtration. The difference from the vertex case
is that the local block is the join $e*A_e^w(t)$ and attaches along
$(\partial e)*A_e^w(t)$.

\begin{theorem}[Local-to-global criterion for flag filtrations]
\label{thm:flag-local-global}
Let $w'$ agree with $w$ away from $e$, and let $w'(e)=s$. If
$$
\Oc_{e,q}^w(t)=0
\qquad
\text{for every $q$ and every }
t\in[\min\{x_e,s\},\max\{x_e,s\}),
$$
then $H_q(\mathcal F_w)\cong H_q(\mathcal F_{w'})$ as persistence modules
for every $q$.
\end{theorem}

\begin{proof}[Proof sketch]
Assume $x_e<s$. For $x_e\leq t<s$, write $A_t=A_e^w(t)$. The two levels
differ by the block $e*A_t$, attached along $(\partial e)*A_t$. Hence
excision and \Cref{lem:join-pair} give
$$
H_q(K_t^w,K_t^{w'})
\cong H_q\bigl(e*A_t,(\partial e)*A_t\bigr)
\cong\redH_{q-2}(A_t)=0.
$$
Outside this interval the levels agree. Naturality then turns the levelwise
isomorphisms into an isomorphism of persistence modules. The reversed case and
the full naturality argument are given in \Cref{app:flag-local-global-proof}.
\end{proof}

\paragraph{Moving one generator.}
Both criteria combine the cone or join attachment lemma of
\Cref{sec:preliminaries} with the persistence equivalence theorem
\cite{EdelsbrunnerHarer2010}: levelwise isomorphisms that commute with the
structure maps give isomorphic persistence modules. What they add is a reading
of the obstruction module. The zero set of $\Oc_\gamma$ around $x_\gamma$ is a
range over which the value of $\gamma$ may move, with all other values fixed,
and leave persistent homology unchanged. Nigmetov and Morozov
\cite{NigmetovMorozov2024} move many values at once, guided by the persistence
pairing of the whole filtration, in order to shift chosen points of the
diagram; here one generator moves, the diagram does not, and the range is read
from a single filtered link. \Cref{sec:diagrams} develops this reading.

\paragraph{Common notation.}
From now on, $\gamma$ denotes either a lower-star vertex or an edge of an
edge-filtered flag complex. We write $x_\gamma$ for its filtration value,
$A_\gamma(t)$ for its filtered link, and
$$
\Oc_{\gamma,q}(t)=\redH_q(A_\gamma(t))
$$
for its obstruction module. The notation always refers to one of the two
constructions above.

%% file: sections/deletion_certificates.tex
\section{Deletion certificates}
\label{sec:certificate}

The preceding section studied how far a generator may move. We now ask when
it may be removed altogether. Regularity at the arrival time is not enough:
the generator may become homologically relevant later. The correct object is
therefore the part of its obstruction module at and after its arrival.

\subsection{Future obstruction}

\begin{definition}[Future obstruction module]
\label{def:future-obstruction}
For a generator $\gamma$ with value $x_\gamma$, define
$$
\Oc_{\gamma,q}^{\geq x_\gamma}(t)
=
\begin{cases}
0, & t<x_\gamma,\\
\Oc_{\gamma,q}(t), & t\geq x_\gamma.
\end{cases}
$$
This is the \emph{future obstruction module} of $\gamma$ in degree $q$.
\end{definition}

Fix a maximum global homological degree $p$. By the degree shift proved below,
only the local degrees $-1\leq q\leq p-\dim\gamma$ are needed to certify
global persistence through degree $p$. We define
$$
\rho_p(\gamma)
=\sup\{d-b:[b,d)\in
\Dgm_q(\Oc_\gamma^{\geq x_\gamma}),
-1\leq q\leq p-\dim\gamma\},
$$
with $\rho_p(\gamma)=0$ when these barcodes are empty. When $p$ is fixed, we
write simply $\rho(\gamma)$. Infinite bars give $\rho(\gamma)=\infty$.

We call $\rho_p(\gamma)$ the \emph{cost} of deleting $\gamma$. By a
\emph{deletion certificate} we mean a condition that is checked in the filtered
link of $\gamma$ alone and that bounds the effect of deleting $\gamma$ on
persistent homology. Ours is the inequality $\rho_p(\gamma)\leq\eps$ for a
chosen tolerance $\eps\geq0$; \Cref{cor:trimming} below shows that it bounds
the bottleneck error by $\eps$. The domination test of strong and edge
collapses is a certificate of the same kind, exact and combinatorial.

For an edge $e=uv$, the entire certificate is computed in the common
neighbourhood of $u$ and $v$. A common neighbour $a$ enters its filtered link
at
$$
\alpha(a)=\max\{w(ua),w(va)\},
$$
and a link edge $ab$ enters at
$$
\beta(ab)=\max\{\alpha(a),\alpha(b),w(ab)\}.
$$
Higher-dimensional link simplices enter by the flag rule. Thus the certificate
never requires construction of the full ambient flag complex.

For an edge and a global degree bound $p$, the relevant local degrees are
$-1\leq q\leq p-1$. In the exact case $\eps=0$, the certificate for $p=1$
therefore asks that the future link be nonempty and connected at every
$t\geq x_e$; for $p=2$ it also asks that the first reduced homology vanish.
These conditions are sufficient, not necessary: in the boundary of an
octahedron, taken as a static complex, deleting an edge preserves $H_0$ and
$H_1$, although its link consists of two points.

\subsection{The relative module of a deletion}

Let $\mathcal F$ be the ambient filtration, with level $K_t$. Deleting
$\gamma$ means deleting it together with every coface. We denote the resulting
filtration by $\mathcal F\setminus\gamma$; its level at $t$ is
$$
K_t\setminus\gamma
=\{\sigma\in K_t:\gamma\not\subseteq\sigma\}.
$$
The relative persistence module of the pair
$(\mathcal F,\mathcal F\setminus\gamma)$ measures exactly what this deletion
removes.

\begin{theorem}[Future obstruction and deletion]
\label{thm:future-relative}
For every $t\geq x_\gamma$ there are natural isomorphisms
$$
H_k(K_t,K_t\setminus\gamma)
\cong
\redH_{k-\dim\gamma-1}(A_\gamma(t)),
$$
while $H_k(K_t,K_t\setminus\gamma)=0$ for $t<x_\gamma$. Consequently,
$$
H_k(\mathcal F,\mathcal F\setminus\gamma)
\cong
\Oc_{\gamma,k-\dim\gamma-1}^{\geq x_\gamma}
$$
as persistence modules.
\end{theorem}

\begin{proof}
Before $x_\gamma$, no simplex containing $\gamma$ is present. For
$t\geq x_\gamma$, the deleted part is the block
$\gamma*A_\gamma(t)$, and it meets $K_t\setminus\gamma$ along
$(\partial\gamma)*A_\gamma(t)$. Excision and
\Cref{lem:join-pair} therefore give
$$
H_k(K_t,K_t\setminus\gamma)
\cong
H_k\bigl(\gamma*A_\gamma(t),(\partial\gamma)*A_\gamma(t)\bigr)
\cong
\redH_{k-\dim\gamma-1}(A_\gamma(t)).
$$
All identifications are induced at the chain level and commute with the
inclusions from one parameter to the next. They therefore assemble into the
stated isomorphism of persistence modules.
\end{proof}

The deletion bound now follows directly from the relative-persistence lemmas
of \Cref{sec:preliminaries}.

\begin{corollary}[Exact and approximate deletion]
\label{cor:trimming}
Suppose every bar of
$\Oc_{\gamma,q}^{\geq x_\gamma}$ has length at most $\eps$ for
$-1\leq q\leq p-\dim\gamma$. Then, for every $0\leq k\leq p$,
$$
d_B\bigl(
\Dgm_k(\mathcal F),
\Dgm_k(\mathcal F\setminus\gamma)
\bigr)
\leq\eps.
$$
In particular,
$$
d_B\bigl(
\Dgm_k(\mathcal F),
\Dgm_k(\mathcal F\setminus\gamma)
\bigr)
\leq\rho_p(\gamma).
$$
If the relevant future obstruction modules vanish, deleting $\gamma$
preserves persistent homology exactly through degree $p$.
\end{corollary}

\begin{proof}
For a fixed $k\leq p$, \Cref{thm:future-relative} identifies the two relative
modules adjacent to
$H_k(\mathcal F\setminus\gamma)\to H_k(\mathcal F)$ with the future
obstruction modules in degrees $k-\dim\gamma-1$ and
$k-\dim\gamma$. The claim follows from
\Cref{lem:zero-relative-module,lem:approx-relative-module}.
\end{proof}

\paragraph{Which local degrees control which global ones.}
The proof works one degree at a time, and only two local degrees enter each
global degree: the bound on $\Dgm_k$ uses the bars of
$\Oc_{\gamma,q}^{\geq x_\gamma}$ for $q=k-\dim\gamma-1$ and
$q=k-\dim\gamma$, and nothing else. Both are needed: vanishing in one of them alone
makes the map $H_k(\mathcal F\setminus\gamma)\to H_k(\mathcal F)$ injective,
or else surjective, which by itself is not an isomorphism. Read the other way, one local degree $q$
is relevant to the two global degrees $q+\dim\gamma$ and $q+\dim\gamma+1$,
and degrees below $-1$ contribute nothing. A single global degree can
therefore be certified on its own, and the range $-1\leq q\leq p-\dim\gamma$
used by $\rho_p(\gamma)$ is the union of these pairs over $k\leq p$.

\paragraph{Why the whole future matters.}
A generator may be regular when it appears and become part of a nontrivial
cycle later. Such a change makes
$H_k(K_t,K_t\setminus\gamma)$ nonzero at the later parameter, and
\Cref{thm:future-relative} records it as a future obstruction bar. The test is
therefore not a one-time check at $x_\gamma$.

The guarantee is also independent of the persistence algorithm used after
preprocessing. It compares the persistence diagrams of two filtrations, not
the reduction orders used to compute them. Any correct backend, greedy or
otherwise, computes the same diagrams and inherits the same bound.

\paragraph{Relation to domination.}
If $A_\gamma(t)$ is a cone for every $t\geq x_\gamma$, then its reduced
homology vanishes and the deletion is exact. The domination tests with which
strong and edge collapses certify a removal
\cite{BoissonnatPritamPareek2018,BoissonnatPritam2019,BoissonnatPritam2020,GlissePritam2022}
are therefore special cases of \Cref{cor:trimming}. The homological
certificate is strictly broader. It admits acyclic links that are not cones:
a link that is a path with three edges, for instance, has no apex. And short
obstruction bars give controlled approximate deletions, as in the running
example of \Cref{fig:running}.

%% file: sections/trimming_algorithm.tex
\section{Sequential trimming}
\label{sec:algorithm}

The deletion certificate is local, but it is not static. Deleting one
generator can change the filtered links of nearby generators, so certificates
computed at the start of a run may later become invalid. We therefore test
each candidate in the current filtration and revisit only candidates whose
filtered links can have changed. The first task is to say exactly which
candidates these are.

\subsection{Algorithm and guarantee}

Fix a maximum homological degree $p$ and a tolerance $\eps\geq0$. Whenever a
candidate $\gamma$ is tested, let $\rho_p^{\mathcal F'}(\gamma)$ denote the
longest relevant bar in its future obstruction module, computed in the current
filtration $\mathcal F'$. The superscript is omitted when the current
filtration is clear.

\paragraph{Which certificates a deletion can change.}
Two distinct generators \emph{conflict} in a complex if some simplex contains
both. For edges of a flag complex, $e$ and $f$ conflict exactly when their
endpoints together span a clique: a triangle if they share a vertex, and a
$4$-clique otherwise. For vertices of a lower-star filtration, two vertices
conflict exactly when they are adjacent. Given a candidate set $\Gamma$, the
\emph{conflict neighbourhood} $C(\gamma)$ is the set of candidates that
conflict with $\gamma$ in the current complex.

\begin{lemma}[Locality of certificates]
\label{lem:locality}
If $\gamma$ and $\delta$ do not conflict in $K$, then the filtered link of
$\delta$ in $\mathcal F\setminus\gamma$ equals its filtered link in
$\mathcal F$. In particular, deleting $\gamma$ changes neither
$\Oc_\delta^{\geq x_\delta}$ nor $\rho_p(\delta)$.
\end{lemma}

\begin{proof}
Let $\tau\in\lk_K(\delta)$. The simplex $\tau\cup\delta$ lies in $K$, and it
does not contain $\gamma$ because $\gamma$ and $\delta$ do not conflict. Hence
$\tau\cup\delta$ survives the deletion and $\tau\in\lk_{K\setminus\gamma}(\delta)$.
The reverse inclusion holds because $K\setminus\gamma\subseteq K$. Every
simplex keeps its filtration value, so the two filtered links agree.
\end{proof}

After a deletion, therefore, only the members of $C(\gamma)$ need to be tested
again. For them a retest is worthwhile: the cost is not monotone under
deletion, since removing simplices from a link can create acyclicity as easily
as destroy it.

\Cref{alg:future-trim} puts these pieces together. The worklist starts as
$\Gamma$ in the order $\pi$, which affects which generators are removed but not
the guarantee; we return to that choice below. Each step takes the next
candidate $\gamma$ and computes its cost in the current filtration, from the
filtered link of $\gamma$ alone. If the cost is at most $\eps$, then $\gamma$
and its cofaces are deleted, the cost is added to the running total $\eta$, and
the candidates of $C(\gamma)$ go back on the worklist; by \Cref{lem:locality}
no others need retesting. Otherwise $\gamma$ stays, to be tested again only if
a later deletion changes its link. The run ends when the worklist is empty, and
returns $\mathcal F'$, $S$ and $\eta$.

For $\eps=0$ every accepted deletion is exact and $\eta=0$. For $\eps>0$ the
total can be conservative: errors from different deletions need not accumulate
in the worst possible way.

\begin{algorithm}[ht]
\caption{\textsc{HomTrim}}
\label{alg:future-trim}
\KwIn{A finite filtration $\mathcal F$; a candidate set $\Gamma$ of generators,
no one of which is a face of another---for us, the vertices of a lower-star
filtration or the edges of an edge-filtered flag complex; a degree bound $p$; a
tolerance $\eps\geq0$; and a processing order $\pi$.}
\KwOut{A trimmed filtration $\mathcal F'$, the deleted set $S$, and the error
bound $\eta$.}
$\mathcal F'\gets\mathcal F$, $S\gets\varnothing$, $\eta\gets0$.\par
Initialize a worklist $Q\gets\Gamma$, ordered by $\pi$.\par
\While{$Q\neq\varnothing$}{
  Extract the next candidate $\gamma$ from $Q$.\par
  Set $r\gets\rho_p^{\mathcal F'}(\gamma)$, computed from
  $\Oc_{\gamma,q}^{\geq x_\gamma}$ for
  $-1\leq q\leq p-\dim\gamma$.\par
  \If{$r\leq\eps$}{
    Record the conflict neighbourhood $C(\gamma)$ in $\mathcal F'$.\par
    Delete $\gamma$ and all its cofaces from $\mathcal F'$.\par
    $S\gets S\cup\{\gamma\}$ and
    $\eta\gets\eta+r$.\par
    Insert into $Q$ the members of $C(\gamma)$ not already in $Q$.\par
  }
}
\Return{$\mathcal F'$, $S$, $\eta$.}
\end{algorithm}

\begin{theorem}[Sequential trimming guarantee]
\label{thm:sequential-trimming}
Let $\mathcal F'$, $S$, and $\eta$ be the output of
\Cref{alg:future-trim}. Then, for every $0\leq k\leq p$,
$$
d_B\bigl(\Dgm_k(\mathcal F),\Dgm_k(\mathcal F')\bigr)
\leq\eta.
$$
In particular, when $\eps=0$, the algorithm preserves persistent homology
through degree $p$.
\end{theorem}

\begin{proof}
Suppose the algorithm deletes $\gamma_1,\ldots,\gamma_m$, and let
$\mathcal F_0=\mathcal F,\mathcal F_1,\ldots,\mathcal F_m=\mathcal F'$ be the
successive filtrations. The certificate for $\gamma_i$ is computed in
$\mathcal F_{i-1}$, so \Cref{cor:trimming} gives
$$
d_B\bigl(\Dgm_k(\mathcal F_{i-1}),\Dgm_k(\mathcal F_i)\bigr)
\leq\rho_p^{\mathcal F_{i-1}}(\gamma_i)
$$
for every $k\leq p$. The triangle inequality now gives
$$
d_B\bigl(\Dgm_k(\mathcal F),\Dgm_k(\mathcal F')\bigr)
\leq
\sum_{i=1}^m\rho_p^{\mathcal F_{i-1}}(\gamma_i)=\eta.
$$
\end{proof}

\begin{proposition}[Termination]
\label{prop:termination}
For a finite filtration and a finite candidate set,
\Cref{alg:future-trim} terminates.
\end{proposition}

\begin{proof}
Every iteration removes one entry from $Q$. Entries are added only at the
start, at most $|\Gamma|$ of them, and after a deletion, at most $|\Gamma|$ at
a time. Each candidate is deleted at most once, so at most
$|\Gamma|+|\Gamma|^2$ entries are ever added, and the loop stops after at most
that many iterations.
\end{proof}

When the run ends, every surviving candidate has been tested since the last
deletion that could have changed its link, since any such deletion would have
put it back on the worklist. By \Cref{lem:locality} neither its link nor its
cost has changed since that test, and that cost is greater than $\eps$: no
further deletion is certified at this tolerance. The claim concerns this
certificate only. A surviving candidate may still be removable by another
method, or by ours at a larger tolerance.

Which candidates survive depends on the order $\pi$, because costs are not
monotone: an early rejection can turn into a deletion once a neighbour is gone,
and an early deletion can raise a neighbour's cost and block it. We use reverse filtration order, as in \cite{GlissePritam2022}, since removing
later generators first tends to leave intact the earlier local structure that
certifies them; no optimality is claimed. \Cref{sec:batch} returns to the question of order.

\subsection{Local cost}

For a candidate $\gamma$, let $T_\gamma^{\mathrm{build}}$ be the time needed
to construct its current filtered link. Since the highest required
obstruction degree is $p-\dim\gamma$, it is enough to build the local chain
complex through dimension $p-\dim\gamma+1$. Let $N_\gamma$ be the number of
simplices of the link of $\gamma$ up to that dimension. This is usually much
smaller than the link itself, whose dimension is not bounded by $p$. A
standard matrix reduction then has the worst-case bounds
$$
O(N_\gamma^3)\quad\text{time},
\qquad
O(N_\gamma^2)\quad\text{space}
$$
\cite{EdelsbrunnerLetscherZomorodian2002,ZomorodianCarlsson2005}.
Consequently, one static pass over $\Gamma$ costs
$$
O\left(
  \sum_{\gamma\in\Gamma}T_\gamma^{\mathrm{build}}
  +\sum_{\gamma\in\Gamma}N_\gamma^3
\right),
$$
with working space $O(\max_\gamma N_\gamma^2)$ when certificates are
computed one at a time.

For an edge $e=uv$ in a flag filtration, the local complex is supported on the
common-neighbour graph of $u$ and $v$. If $c_\ell(e)$ is the number of
$\ell$-cliques in this graph, then
$$
N_e=\sum_{\ell=1}^{p+1}c_\ell(e).
$$
Writing $C_e$ for the set of common neighbours, $N_e=O(|C_e|^{p+1})$. For
fixed $p$ the cost is therefore polynomial in the size of the common
neighbourhood, and it does not depend on the size of the ambient complex.

In the sequential algorithm, let $n_\gamma^{\mathrm{test}}$ be the number of
times $\gamma$ is tested, and take $T_\gamma^{\mathrm{build}}$ and
$N_\gamma$ to be upper bounds over those tests. The corresponding bound is
$$
O\left(
  \sum_\gamma n_\gamma^{\mathrm{test}}T_\gamma^{\mathrm{build}}
  +\sum_\gamma n_\gamma^{\mathrm{test}}N_\gamma^3
\right).
$$
The number of tests is also local. Complexes only shrink during a run, so two
candidates that do not conflict in the input never conflict later. A
candidate is reinserted only when a member of its conflict neighbourhood is
deleted, and each member is deleted at most once. Hence
$n_\gamma^{\mathrm{test}}\leq1+\deg_c(\gamma)$, where the \emph{conflict
degree} $\deg_c(\gamma)$ is the size of $C(\gamma)$ in the input. For an
edge $e=uv$ of a flag complex, the edges that conflict with $e$ are $ux$ and
$vx$ for $x\in C_e$, together with the edges of $G$ inside $C_e$, so
$$
\deg_c(e)\leq2|C_e|+|E(G[C_e])|,
$$
with equality when every edge is a candidate. This count uses only the
vertices and edges of the link of $e$, which the certificate builds anyway.

\subsection{Implementation}

The implementation avoids a full local persistence reduction when a cheaper
test is available.

\paragraph{Domination first.}
If the future filtered link remains a cone, then its reduced homology
vanishes. We test this combinatorial certificate first because it is cheaper
than computing local homology
\cite{BoissonnatPritamPareek2018,BoissonnatPritam2020}.

\paragraph{Tests by degree.}
For $p=1$, an edge certificate asks only whether the future link is nonempty
and connected, so graph searches or union--find suffice. For $p=2$, the new
condition is the vanishing of $\redH_1$; for $p=3$, one must also test
$\redH_2$. Our implementation processes the finitely many local event values
in order and updates the relevant boundary ranks, rather than rebuilding the
chain complex at every event.

\paragraph{Current-filtration recomputation.}
After a deletion, only candidates whose common-neighbour complexes may have
changed are returned to the worklist. This set contains the conflict
neighbourhood of the deleted edge, which is all that \Cref{lem:locality}
requires. Every accepted certificate is therefore computed in the filtration
in which the deletion is actually performed, as required by
\Cref{thm:sequential-trimming}.

%% file: sections/batch_deletion.tex
\section{Independent families and deletion rounds}
\label{sec:batch}

\Cref{thm:sequential-trimming} charges every deletion separately, so its
certified error grows with the number of deletions of positive cost. The
sequential algorithm also has to recompute certificates, because one deletion
can change the certificate of another. Both effects have the same source: a
simplex that contains two of the deleted generators. This section shows that
when no such simplex exists, a whole family can be deleted at once at the cost
of its most expensive member, and that this reduces the question of deletion
order, for such families, to a colouring problem on a graph with local
degrees.

\subsection{Batch deletion}

\begin{definition}[Independent family]
\label{def:independent}
A set $S$ of generators is \emph{independent} in $K$ if no simplex of $K$
contains two distinct elements of $S$, that is, if no two members of $S$
conflict.
\end{definition}

For edges of a flag complex, $S$ is independent when no two of its edges span
a clique together. For vertices of a lower-star filtration, $S$ is an
independent set of the one-skeleton. Deleting $S$ means deleting every
simplex that contains an element of $S$; we write $K\setminus S$ and
$\mathcal F\setminus S$ for the results. If $S$ is independent in $K$, it is
independent in every $K_t$ and in every subcomplex of $K$.

\begin{lemma}[Splitting]
\label{lem:splitting}
If $S$ is independent in $K$, then there is an isomorphism of chain complexes
$$
C_*(K,K\setminus S)\cong\bigoplus_{\gamma\in S}C_*(K,K\setminus\gamma).
$$
\end{lemma}

\begin{proof}
The relative chains $C_*(K,K\setminus S)$ have a basis consisting of the
simplices of $K$ that contain an element of $S$. By independence, each such
simplex $\sigma$ contains exactly one element, which we denote by $s(\sigma)$.
Grouping the basis by $s(\sigma)$ splits $C_*(K,K\setminus S)$, as a graded
vector space, into summands $V_\gamma$ spanned by the simplices containing
$\gamma$. These are exactly the basis simplices of $C_*(K,K\setminus\gamma)$.

It remains to compare boundaries. Let $s(\sigma)=\gamma$ and let $\tau$ be a
facet of $\sigma$. If $\tau$ contains $\gamma$, then $s(\tau)=\gamma$ and
$\tau\in V_\gamma$. Otherwise $\tau$ contains no element of $S$, because every
element of $S$ in $\tau$ would also lie in $\sigma$. Then $\tau\in
K\setminus S$ and it vanishes in the quotient. The boundary of $\sigma$ in
$C_*(K,K\setminus S)$ is therefore the sum of its facets that contain
$\gamma$, which is its boundary in $C_*(K,K\setminus\gamma)$.
\end{proof}

\begin{theorem}[Batch deletion]
\label{thm:batch}
Let $S$ be independent in $K$. Then, for every $k$,
$$
H_k(\mathcal F,\mathcal F\setminus S)
\cong
\bigoplus_{\gamma\in S}\Oc_{\gamma,k-\dim\gamma-1}^{\geq x_\gamma}
$$
as persistence modules. Consequently, for every $0\leq k\leq p$,
$$
d_B\bigl(\Dgm_k(\mathcal F),\Dgm_k(\mathcal F\setminus S)\bigr)
\leq\max_{\gamma\in S}\rho_p(\gamma).
$$
In particular, if every member of $S$ has cost zero, deleting $S$ preserves
persistent homology exactly through degree $p$.
\end{theorem}

\begin{proof}
Apply \Cref{lem:splitting} to every level $K_t$. The assignment
$\sigma\mapsto s(\sigma)$ does not depend on $t$, so the levelwise
isomorphisms commute with the inclusions and give an isomorphism of relative
persistence modules
$H_k(\mathcal F,\mathcal F\setminus S)\cong
\bigoplus_{\gamma}H_k(\mathcal F,\mathcal F\setminus\gamma)$.
\Cref{thm:future-relative} identifies each summand. A structure map of a
direct sum is the direct sum of the structure maps, so if each summand is
$\eps_\gamma$-trivial, the sum is $(\max_\gamma\eps_\gamma)$-trivial. With
$\eps_\gamma=\rho_p(\gamma)$, \Cref{lem:approx-relative-module} gives the
bottleneck bound and \Cref{lem:zero-relative-module} the exact case.
\end{proof}

Deleting the members of $S$ one at a time certifies only
$\sum_{\gamma\in S}\rho_p(\gamma)$, although by \Cref{lem:locality} their
costs do not change along the way. The sum in \Cref{thm:sequential-trimming}
is the price of an argument that must allow deletions to interact. For an
independent family they do not, and the maximum is the correct bound.

Independence cannot be dropped. In the flag complex of a triangle $xyz$, with
all edges at value $0$, each of the edges $xy$ and $xz$ has a one-point link
from its arrival on, so each can be deleted exactly. Deleting both disconnects
$x$.

\subsection{Rounds}

A run can be organized as a sequence of batches. The triangle inequality
gives the following bound, exactly as in the proof of
\Cref{thm:sequential-trimming}, with a round in place of a single deletion.

\begin{corollary}[Deletion rounds]
\label{cor:rounds}
Let $\mathcal F^{(0)}=\mathcal F$. For $j=1,\ldots,r$, let $S_j$ be
independent in the final complex of $\mathcal F^{(j-1)}$, with costs computed
in $\mathcal F^{(j-1)}$, and let $\mathcal F^{(j)}=\mathcal F^{(j-1)}\setminus
S_j$. Then, for every $0\leq k\leq p$,
$$
d_B\bigl(\Dgm_k(\mathcal F),\Dgm_k(\mathcal F^{(r)})\bigr)
\leq\sum_{j=1}^r\max_{\gamma\in S_j}\rho_p^{\mathcal F^{(j-1)}}(\gamma).
$$
If every deleted generator has cost at most $\eps$, the right-hand side is
at most $r\eps$.
\end{corollary}

The certified error is now governed by the number of rounds rather than the
number of deletions. Let the \emph{conflict graph} have the candidates as
vertices and an edge between every conflicting pair. An independent family is
an independent set of this graph, and a partition of the candidates into $r$
independent families is a proper colouring with $r$ colours. Greedy colouring
uses at most $\Delta_c+1$ colours, where $\Delta_c$ is the maximum conflict
degree.

\begin{proposition}[One colouring pass]
\label{prop:colouring-pass}
Let $\Gamma_0$ be the candidates of cost at most $\eps$ in $\mathcal F$, and
let $\Delta_c$ be the maximum degree of the conflict graph on $\Gamma_0$.
Colour this graph greedily with classes $S_1,\ldots,S_r$, where
$r\leq\Delta_c+1$. Process the classes in order: in round $j$, recompute the
costs of the surviving members of $S_j$ in the current filtration and delete,
together, those of cost at most $\eps$. The output $\mathcal F'$ satisfies
$$
d_B\bigl(\Dgm_k(\mathcal F),\Dgm_k(\mathcal F')\bigr)
\leq(\Delta_c+1)\,\eps
\qquad\text{for }0\leq k\leq p,
$$
however many generators are deleted.
\end{proposition}

\begin{proof}
A colour class is independent in $K$ and hence in every later complex, and so
is each of its subsets. Every round therefore deletes an independent family
whose costs, computed in the current filtration, are at most $\eps$.
\Cref{cor:rounds} gives the bound $r\eps\leq(\Delta_c+1)\eps$.
\end{proof}

The quantity $\Delta_c$ is local. For an edge $e=uv$ of a flag complex,
\Cref{sec:algorithm} showed that its conflict degree is at most
$2|C_e|+|E(G[C_e])|$, a count of the vertices and edges of its link. For
vertices of a lower-star filtration, the conflict graph is a subgraph of the
one-skeleton, so $\Delta_c$ is at most its maximum degree. For the
Freudenthal triangulation of a regular grid this is $6$ in the plane and $14$
in space, whatever the size of the grid.

A single pass stops when its colour classes are used up. Deletions can create
new candidates, however, because costs are not monotone. The
\emph{exhaustive} round scheme therefore repeats: it computes the costs of the
candidates that may have changed, chooses a maximal independent family among
those of cost at most $\eps$, deletes it, and stops when no candidate
qualifies. \Cref{cor:rounds} bounds its error by $r\eps$, where $r$ is the
number of rounds executed. We have no a priori bound on $r$ for this scheme.

\begin{table}[ht]
\centering
\small
\begin{tabular}{@{}lrrrrr@{}}
\toprule
family & $m$ & deleted & rounds $r$ & round bound & additive bound \\
\midrule
block\_clique, $b=10$ &  270 &  108 & 10 & 0.18 & 0.41 \\
                      &  540 &  230 & 11 & 0.28 & 0.74 \\
                      & 1080 &  498 & 11 & 0.34 & 1.24 \\
                      & 2160 & 1037 & 11 & 0.38 & 2.43 \\
\midrule
geom\_knn3, $k=8$     &  299 &  184 &  7 & 0.24 & 0.90 \\
                      &  588 &  417 &  9 & 0.30 & 1.92 \\
                      & 1135 &  843 & 14 & 0.49 & 3.86 \\
                      & 2296 & 1676 & 10 & 0.42 & 7.06 \\
\midrule
sparse\_matrix, $\bar d=12$ &  360 &  75 & 4 & 0.05 & 0.14 \\
                      &  720 &  155 &  4 & 0.13 & 0.47 \\
                      & 1440 &  368 &  5 & 0.14 & 0.84 \\
                      & 2880 &  873 &  5 & 0.22 & 1.81 \\
\bottomrule
\end{tabular}
\caption{Exhaustive deletion rounds at $\eps=0.05$ and $p=2$ on small
instances of the families of \Cref{sec:experiments}, with the number of
vertices growing eightfold within each family. Here $m$ is the number of
input edges, the round bound is the right-hand side of \Cref{cor:rounds}, and
the additive bound is $\sum\rho_p$ over the same deletions, which is what
\Cref{thm:sequential-trimming} would certify. These runs use an independent
Python prototype built on \textsc{Gudhi}, not the implementation of
\Cref{sec:experiments}.}
\label{tab:rounds}
\end{table}

\Cref{tab:rounds} shows how the two bounds behave in practice. The number of
deletions grows roughly tenfold within each family, while the number of
rounds stays between $4$ and $14$. For block cliques a size-independent
count is expected, since disjoint blocks never conflict; the geometric and
sparse-matrix families show the same behaviour without that structure. The additive bound grows with the
number of deletions, and on the largest instances it exceeds the round bound
by factors of $6$ to $17$. A single colouring pass on the largest instances
deletes fewer edges than the exhaustive scheme ($959$, $1285$, and $765$
against $1037$, $1676$, and $873$) and certifies larger errors ($0.73$,
$0.76$, and $0.24$, against $0.38$, $0.42$, and $0.22$). The single pass has
an a priori guarantee; on these instances the exhaustive scheme did better.

\paragraph{Parallelism.}
Within a round, every certificate is computed against the same complex, and
the deletions remove pairwise disjoint sets of simplices, by the partition in
the proof of \Cref{lem:splitting}. Both steps can therefore be carried out in
parallel, and by \Cref{lem:locality} the effect of each deletion on later
certificates is confined to its conflict neighbourhood. One point needs care
in an implementation: independent edges may share a vertex, since $uv$ and
$uw$ are independent whenever $v$ and $w$ are not adjacent, so concurrent
updates to per-vertex adjacency must be synchronized. The implementation
evaluated in \Cref{sec:experiments} is sequential; whether a parallel round
scheme recovers the cost of preprocessing end to end is a question for
measurement.

\paragraph{The dependent case.}
Outside independent families we know no bound better than the sum in
\Cref{thm:sequential-trimming}. The triangle example shows that the maximum
can fail by an arbitrary amount when the costs are not recomputed. With
recomputation, the prototype found pairs and triples of interacting deletions
whose combined bottleneck error reached $96\%$ and $99\%$ of the sum of their
costs, so the sum can be nearly attained by short dependent sequences. Longer
runs behaved differently: in the $22$ sampled runs with six or more deletions
of positive cost, the error never exceeded $28\%$ of the additive bound. We do
not know whether some family forces the error to grow linearly with the number
of deletions.

%% file: sections/regularity_descriptors.tex
\section{Regularity descriptors}
\label{sec:diagrams}

We now return to the first reading of the obstruction module, from
\Cref{sec:obstruction}. The deletion certificate uses the module from
$x_\gamma$ onward. The whole module, on both sides of $x_\gamma$, also says
how far the generator can move before it becomes visible to homology. This
section turns that information into a pair of slacks for each generator and
asks what remains stable under perturbation. None of the algorithmic results
depend on it.

\subsection{Regularity intervals and slack}

Fix a generator $\gamma$ with value $x_\gamma$. Its \emph{bad set} is the set
of parameters at which some obstruction group is nonzero:
$$
B_\gamma
=\{t\in\mathbb R:\Oc_{\gamma,q}(t)\neq0
  \text{ for some }q\}.
$$
Equivalently, $B_\gamma$ is the union of the supports of the bars in the
graded obstruction barcode. The generator is regular at its anchor precisely
when $x_\gamma\notin B_\gamma$.

\begin{definition}[Regularity interval and slack]
\label{def:regularity-interval}
Suppose $\gamma$ is regular at $x_\gamma$. Its \emph{regularity interval}
$I_\gamma$ is the connected component of $\mathbb R\setminus B_\gamma$ that
contains $x_\gamma$. Let $\ell_\gamma$ and $r_\gamma$ be its left and right
endpoints in the extended real line. The \emph{backward slack}, \emph{forward
slack}, and \emph{margin} of $\gamma$ are
$$
L_\gamma=x_\gamma-\ell_\gamma,
\qquad
R_\gamma=r_\gamma-x_\gamma,
\qquad
m_\gamma=\min\{L_\gamma,R_\gamma\}.
$$
\end{definition}

This definition avoids an unnecessary endpoint convention: the obstruction
bars are half-open, but only the two endpoint values enter the descriptors.
Infinite slack is allowed. For example, $R_\gamma=\infty$ when no obstruction
appears after $x_\gamma$.

\begin{example}
\label{ex:sublevel-vertex}
Suppose a vertex $v$ has value $x_v=5$, and that its obstruction module
vanishes for $2\leq t<9$ but not just below $2$ or at $9$. Then
$I_v=[2,9)$, $L_v=3$, $R_v=4$, and $m_v=3$. By \Cref{thm:ls-local-global},
moving the value of $v$ anywhere in $[2,9]$, while keeping all other values
fixed, leaves persistent homology unchanged (\Cref{fig:rsd}).
\end{example}

Thus the margin is an exact one-generator perturbation budget. It does not
assert that all generators may be moved simultaneously by their margins: such
a perturbation changes several local filtrations at once.

In computations one restricts to the local degrees $-1\leq q\leq p-\dim\gamma$,
as for the deletion certificate. The resulting truncated descriptors carry the
same guarantee through global degree $p$, and the stability theorems below
hold for them with the same proofs.

\subsection{The regularity diagrams}

We collect the endpoints, or equivalently the two slacks, over all regular
generators.

\begin{definition}[RPD and RSD]
\label{def:rpd-rsd}
For a filtration $\mathcal F$, its \emph{regularity persistence diagram} and
\emph{regularity slack diagram} are the multisets
$$
\RPD(\mathcal F)
=\{(\ell_\gamma,r_\gamma):\gamma\text{ is regular}\}
  \cup\Delta^\infty,
$$
and
$$
\RSD(\mathcal F)
=\{(L_\gamma,R_\gamma):\gamma\text{ is regular}\}
  \cup\partial^\infty.
$$
Here $\Delta^\infty$ is the diagonal in the extended plane and
$\partial^\infty$ is the boundary of $\mathbb R_{\geq0}^2$, each included
with infinite multiplicity.
\end{definition}

The RPD records where each regular window lies on the filtration line. The RSD
recentres the same window at the generator and records its freedom in the two
directions. The $\ell_\infty$-distance from $(L_\gamma,R_\gamma)$ to the
boundary is exactly $m_\gamma$. Critical generators contribute no interior
point and are represented by the boundary. We write $d_B^\partial$ for
bottleneck distance in the closed quadrant, with interior points allowed to
match to its boundary.

These diagrams are not determined by the ordinary persistence diagram. Let $K$
be a single edge $uv$ with the lower-star function $f(u)=0$ and $f(v)=s>0$.
For every $s$, the persistence diagram consists of one essential class born
at $0$. The vertex $u$ is critical, since its link is empty before $s$, while
$v$ is regular with $I_v=[0,\infty)$. The RSD is therefore the single point
$(s,\infty)$, which records where $v$ sits inside the window in which it can
move freely.

\input{sections/rsd_figure}

\subsection{Stability}
\label{sec:stability}

The obstruction barcode is stable, but an exact regularity interval need not
be. A short bar can appear inside a long zero window and split it, even though
that bar is inexpensive to match to the diagonal. This distinction is the
reason for separating exact slack from its robust, multiscale version. The
first statement is the standard stability theorem, applied to each local
filtration; we record it because everything else is measured against it.

\begin{theorem}[Stability of the obstruction module]
\label{thm:universal-obstruction-stability}
Let $f$ and $g$ be two lower-star functions on the same complex, or two edge
functions on the same graph, with $\|f-g\|_\infty\leq\eps$. Then, for every
generator $\gamma$, the obstruction modules $\Oc_\gamma^f$ and
$\Oc_\gamma^g$ are $\eps$-interleaved. Consequently,
$$
d_B\bigl(\Dgm_q(\Oc_\gamma^f),\Dgm_q(\Oc_\gamma^g)\bigr)\leq\eps
$$
for every degree $q$.
\end{theorem}

The corresponding generatorwise statement is false for exact slack.

\begin{proposition}[Exact slack is not Lipschitz]
\label{prop:no-go}
Already for edge-filtered flag complexes, there is no constant $C$ for which
$$
\max\{|\ell_\gamma^f-\ell_\gamma^g|,
       |r_\gamma^f-r_\gamma^g|,
       |L_\gamma^f-L_\gamma^g|,
       |R_\gamma^f-R_\gamma^g|\}
\leq C\|f-g\|_\infty
$$
holds for every pair of edge functions $f,g$ and every generator $\gamma$
that is regular in both filtrations.
\end{proposition}

Exact slack is stable when the geometry of the bad sets is controlled. For a
function $f$, centre the bad set at the generator by writing
$$
C_\gamma^f
=\overline{\{t-x_\gamma^f:t\in B_\gamma^f\}}.
$$

\begin{theorem}[Bad-set control implies slack stability]
\label{thm:badset-slack-stability}
Suppose the filtrations have the same generators and
$d_H(C_\gamma^f,C_\gamma^g)\leq\delta$ for every $\gamma$, where $d_H$
denotes Hausdorff distance. Then
$$
d_B^\partial\bigl(\RSD(\mathcal F_f),\RSD(\mathcal F_g)\bigr)
\leq\delta.
$$
\end{theorem}

For an unconditional statement, we introduce a scale. Given $\tau\geq0$, let
$B_\gamma^{f,\tau}$ be the union of the supports of those obstruction bars
whose length is greater than $\tau$. The component of its complement
containing $x_\gamma^f$ defines the $\tau$-robust regularity interval, whenever
the anchor is not itself in $B_\gamma^{f,\tau}$. Its endpoints and slacks give
the robust diagrams $\RPD_\tau$ and $\RSD_\tau$. We write $N_\eps(S)$ for the
closed $\eps$-neighbourhood of a set $S\subseteq\mathbb R$.

\begin{theorem}[Multiscale stability of robust obstruction sets]
\label{thm:robust-badset-stability}
If $\|f-g\|_\infty\leq\eps$, then for every generator $\gamma$ and every
$\tau\geq0$,
$$
B_\gamma^{f,\tau+2\eps}
\subseteq N_\eps(B_\gamma^{g,\tau}),
\qquad
B_\gamma^{g,\tau+2\eps}
\subseteq N_\eps(B_\gamma^{f,\tau}).
$$
After centring at the two anchor values, the neighbourhood radius is
$2\eps$.
\end{theorem}

The theorem gives a scale-shifted form of stability: an obstruction that
survives the larger threshold $\tau+2\eps$ must remain visible, up to an
$\eps$ displacement, at threshold $\tau$ after perturbation. It does not
claim a same-scale generatorwise bound for exact slack, which
\Cref{prop:no-go} rules out. Proofs of all statements in this subsection are
given in \Cref{app:stability-proofs}; the endpoint degrees of a regularity
interval are discussed separately in \Cref{app:endpoint-bigrading}.

%% file: sections/rsd_figure.tex
\begin{figure}[t]
\centering
\begin{tikzpicture}[
  x=1cm,y=1cm,>=Latex,
  panel/.style={font=\small\bfseries},
  lbl/.style={font=\footnotesize},
  slbl/.style={font=\scriptsize}
]

\begin{scope}[shift={(0,0)}]
  \node[panel,anchor=south west] at (-0.55,2.05) {(a)};

  \def\s{0.60}
  \draw[->,black!70] (0,0) -- ({11.6*\s},0);
  \node[lbl,anchor=west,inner sep=3pt] at ({11.6*\s},0) {$t$};

  \draw[line width=2.6pt,red!70!black] ({0.4*\s},0) -- ({2*\s},0);
  \draw[line width=2.6pt,red!70!black] ({9*\s},0) -- ({10.8*\s},0);
  \node[slbl,red!70!black,anchor=south] at ({1.2*\s},0.10) {bad};
  \node[slbl,red!70!black,anchor=south] at ({9.9*\s},0.10) {bad};

  \draw[black!70] ({2*\s},0.09) -- ({2*\s},-0.09);
  \draw[black!70] ({9*\s},0.09) -- ({9*\s},-0.09);
  \node[slbl,anchor=north,inner sep=3pt] at ({2*\s},-0.11) {$\ell_v=2$};
  \node[slbl,anchor=north,inner sep=3pt] at ({9*\s},-0.11) {$r_v=9$};

  \filldraw[blue!65!black] ({5*\s},0) circle (2.1pt);
  \node[slbl,blue!65!black,anchor=south,inner sep=3pt] at ({5*\s},0.06) {$x_v=5$};

  \draw[decorate,decoration={brace,amplitude=5pt},blue!55!black,line width=0.7pt]
    ({2*\s},0.62) -- ({9*\s},0.62);
  \node[lbl,blue!55!black,anchor=south,inner sep=5pt] at ({5.5*\s},0.66)
    {regularity interval $I_v$};

  \draw[<->,black!80,line width=0.7pt] ({2*\s},-0.72) -- ({5*\s},-0.72);
  \draw[<->,black!80,line width=0.7pt] ({5*\s},-0.72) -- ({9*\s},-0.72);
  \draw[black!35,dash pattern=on 1.3pt off 1.5pt] ({5*\s},-0.66) -- ({5*\s},-0.06);
  \node[slbl,anchor=north,inner sep=3pt] at ({3.5*\s},-0.75) {$L_v=3$};
  \node[slbl,anchor=north,inner sep=3pt] at ({7*\s},-0.75) {$R_v=4$};
\end{scope}

\begin{scope}[shift={(8.55,-1.15)}]
  \node[panel,anchor=south west] at (-1.05,3.20) {(b)};

  \draw[red!70!black,line width=1.4pt,-{Latex[length=2mm]}] (0,0) -- (0,3.25);
  \draw[red!70!black,line width=1.4pt,-{Latex[length=2mm]}] (0,0) -- (4.55,0);
  \node[lbl,anchor=east,inner sep=4pt]  at (0,3.20) {$R$};
  \node[lbl,anchor=north,inner sep=4pt] at (4.50,0) {$L$};

  \node[slbl,red!70!black,anchor=north west,align=left,inner sep=2pt]
    at (0.10,-0.16) {critical boundary $\partial^\infty$};

  \draw[black!45,dash pattern=on 1.6pt off 1.6pt] (0.95,1.38) rectangle (2.05,2.48);
  \filldraw[blue!65!black] (1.50,1.93) circle (2.1pt);
  \node[slbl,blue!65!black,anchor=west,inner sep=4pt] at (1.56,1.95) {$v$};
  \node[slbl,blue!65!black,anchor=north,inner sep=4pt] at (1.50,1.33)
    {$(L_v,R_v)=(3,4)$};
  \node[slbl,black!65,anchor=west,inner sep=4pt] at (2.10,1.93)
    {$\ell_\infty$-ball, radius $\delta$};

  \filldraw[blue!65!black] (3.60,0.60) circle (2.1pt);
  \node[slbl,blue!65!black,anchor=west,inner sep=4pt] at (3.68,0.60) {$\gamma'$};

  \filldraw[blue!65!black] (0.45,3.02) circle (2.1pt);
  \node[slbl,blue!65!black,anchor=west,inner sep=4pt] at (0.52,3.02)
    {$\gamma''$, near-critical};
  \draw[{Latex[length=1.3mm]}-{Latex[length=1.3mm]},black!80,line width=0.6pt]
    (0,2.74) -- (0.45,2.74);
  \node[slbl,anchor=west,inner sep=2pt] at (0.50,2.62) {$m_{\gamma''}$};
\end{scope}

\end{tikzpicture}
\caption{Regularity around one generator. (a) The obstruction module of the
vertex $v$ vanishes between the two red parts of its bad set. The anchor
$x_v=5$ divides this regularity interval into slacks $L_v=3$ and $R_v=4$.
(b) The corresponding point $(L_v,R_v)$ in the regularity slack diagram. Its
$\ell_\infty$-distance to the boundary is the margin $m_v$; points sufficiently
close to the boundary may instead be matched to it.}
\label{fig:rsd}
\end{figure}
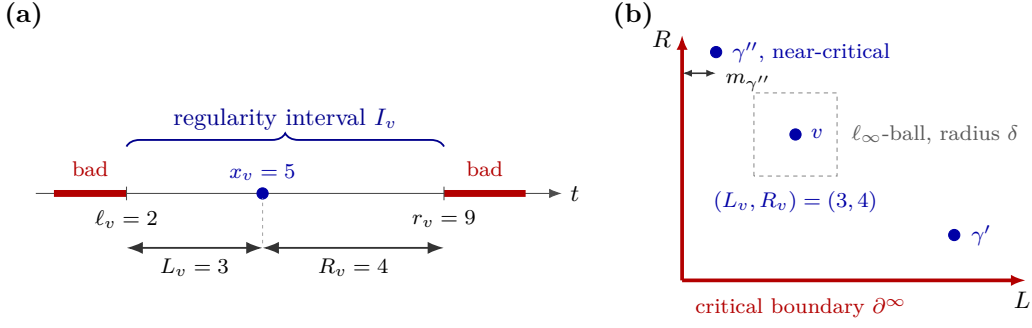

%% file: sections/experiments.tex
\section{Experiments}
\label{sec:experiments}

We evaluate \textsc{HomTrim} as a preprocessing step for weighted flag
filtrations. An input is a graph $G=(V,E)$ with edge function
$w:E\to\mathbb R$; an edge appears at its assigned value, and each clique
appears at the largest value of one of its edges. The parameter $p$ is the
largest homological degree certified by the preprocessor and computed by the
persistence backend.

\paragraph{Pipelines and measurements.}
We compare three pipelines. \textsc{Raw} runs Ripser directly on the input
\cite{Bauer2021}. \textsc{DomTrim} applies domination-based trimming before
Ripser. \textsc{HomTrim} first applies the same domination test and then the
homological certificate of \Cref{sec:certificate}.

For each pipeline, we distinguish preprocessing time, Ripser time, and their
sum. We write $m$, $m_D$, and $m_H$ for the numbers of input edges, edges after
\textsc{DomTrim}, and edges after \textsc{HomTrim}. All runs use $\eps=0$, so
every deletion is exact and the trimmed filtration has the same persistence
diagrams as the input through degree $p$. The certificates establish
correctness; direct barcode comparisons are used only as implementation
checks. The implementation is the sequential \Cref{alg:future-trim}; the round
scheme of \Cref{sec:batch} is not part of these measurements.

\paragraph{Input families.}
The geometric inputs are symmetrized $k$-nearest-neighbour graphs on uniform
points in $[0,1]^d$, with Euclidean edge weights. The \emph{sparse-matrix}
family plants small random cliques in a sparse weighted graph and fills the
remaining edge budget with random edges. The \emph{block-clique} family is a
disjoint collection of weighted complete graphs on small blocks. It is
artificial, but isolates the locally clique-rich structure on which lower-link
tests are most useful. A case is one generated graph together with its
parameters, value of $p$, and random seed.

\paragraph{Reduction beyond domination.}
\Cref{tab:value-add-compact} answers the main experimental question:
homological trimming consistently removes edges left by domination. The
additional reduction ranges from $12.8\%$ to $40.6\%$ of the original edge
set. Ripser is correspondingly $1.55\times$ to $4.61\times$ faster on the
\textsc{HomTrim} output than on the \textsc{DomTrim} output. These backend
gains do not by themselves give a faster full pipeline. In every row of this
table the last column is below one: the extra local-homology work costs more
than it saves in the backend, so the complete \textsc{HomTrim} pipeline is
slower than the \textsc{DomTrim} pipeline.

\input{tables/value_add_compact}

\paragraph{When preprocessing pays for itself.}
Against direct computation the comparison changes with density.
\Cref{tab:sparsity-compact} compares the complete \textsc{HomTrim} pipeline
with raw Ripser for $p=2$ and $n=10{,}000$. At low average degree, trimming
reduces the graph but its overhead dominates. As the degree grows, the ratio
moves towards trimming in both families: raw Ripser slows down faster than the
trimmed pipeline. For the sparse-matrix
family at average degree $48$, raw Ripser takes $4.79$ seconds and the
complete \textsc{HomTrim} pipeline takes $3.87$ seconds, an end-to-end speedup
of $1.24\times$. On a single instance of degree $64$, raw Ripser times out
while the trimmed pipeline completes in $303.23$ seconds. We report this as an extension of the
computable range, not as a finite speedup ratio. On the geometric family the
pipeline is still slower than raw Ripser at degree $48$, although the gap has
narrowed to $1.70$ against $2.12$ seconds.

\input{tables/sparsity_compact}

\paragraph{Higher degree.}
As $p$ increases, the certificate must rule out more local homology and the
additional reduction decreases. It remains visible at $p=3$, where the three
families in \Cref{tab:value-add-compact} show extra reductions of
$12.8\%$--$33.6\%$. We also ran trim-only stress tests for $p=5,6,7$ on sparse
block-clique and sparse-matrix inputs. Ripser encounters a simplex-index
overflow on these high-dimensional flag filtrations, so the runs are not used
for end-to-end timing. As \Cref{tab:high-p-trim-only} shows, the local
preprocessor itself removes $36.9\%$--$44.7\%$ of the edges.

\input{tables/high_p_trim_only}

\paragraph{Interpretation.}
Homological links remove edges that domination leaves, and the backend runs
correspondingly faster. Whether the complete pipeline is faster depends on the
input. Against direct computation it wins on the denser sparse-matrix inputs
and extends the range of instances that can be computed at all; against
domination-only preprocessing, the local homology is not yet recovered end to
end in our sequential implementation. The present tool is therefore a
certified preprocessor for moderately sparse, locally clique-rich inputs, not
a replacement for direct persistence computation.

A small controlled experiment with the regularity descriptors of
\Cref{sec:diagrams} appears in \Cref{app:rsd-experiment}. It illustrates the
descriptors rather than validating a downstream application, so we keep it
separate from these measurements.

\FloatBarrier

%% file: tables/value_add_compact.tex
\begin{table}[ht]
\centering
\scriptsize
\setlength{\tabcolsep}{3.8pt}
\renewcommand{\arraystretch}{1.0}
\begin{tabular}{llrcccc}
\toprule
Family & \(p\) & cases & \(m\to m_D\to m_H\) & extra & PH sp. & Dom/Hom total \\
\midrule
block\_clique & 2 & 18 & 86,186\(\to\)85,604\(\to\)50,273 & 40.6\% & 2.64\(\times\) & 0.25\(\times\) \\
geom\_knn3 & 2 & 12 & 117,584\(\to\)55,257\(\to\)35,230 & 16.9\% & 2.04\(\times\) & 0.83\(\times\) \\
sparse\_matrix & 2 & 12 & 105,000\(\to\)82,803\(\to\)68,705 & 13.2\% & 1.55\(\times\) & 0.87\(\times\) \\
block\_clique & 3 & 9 & 57,436\(\to\)57,042\(\to\)37,971 & 33.6\% & 4.61\(\times\) & 0.18\(\times\) \\
geom\_knn3 & 3 & 6 & 78,609\(\to\)36,193\(\to\)23,124 & 16.6\% & 2.42\(\times\) & 0.90\(\times\) \\
sparse\_matrix & 3 & 6 & 70,000\(\to\)55,992\(\to\)46,940 & 12.8\% & 1.66\(\times\) & 0.95\(\times\) \\
\bottomrule
\end{tabular}
\caption{Value added by \textsc{HomTrim} over domination-only trimming.
Here $m$, $m_D$, and $m_H$ are the edge counts before trimming, after
\textsc{DomTrim}, and after \textsc{HomTrim}. ``Extra'' is the mean additional
reduction $100(m_D-m_H)/m$. ``PH sp.'' is the ratio of Ripser times
$T_D^{\rm PH}/T_H^{\rm PH}$, and ``Dom/Hom total'' is the ratio of total
pipeline times $T_D^{\rm total}/T_H^{\rm total}$. Thus a total-time ratio below
one means that \textsc{HomTrim} is slower after preprocessing is included.}
\label{tab:value-add-compact}
\end{table}

%% file: tables/sparsity_compact.tex
\begin{table}[ht]
\centering
\scriptsize
\setlength{\tabcolsep}{2.5pt}
\renewcommand{\arraystretch}{0.95}
\begin{tabular}{llrccc}
\toprule
Family & degree & seeds & \(m\to m_D\to m_H\) & extra & raw/Hom total \\
\midrule
geom\_knn3 & 12 & 1 & 69,046\(\to\)37,494\(\to\)32,444 & 7.3\% & 0.13/0.28 \\
geom\_knn3 & 24 & 1 & 134,730\(\to\)67,681\(\to\)45,216 & 16.7\% & 0.47/0.79 \\
geom\_knn3 & 48 & 1 & 266,199\(\to\)103,455\(\to\)57,721 & 17.2\% & 1.70/2.12 \\
sparse\_matrix & 12 & 1 & 60,000\(\to\)42,278\(\to\)37,181 & 8.5\% & 0.12/0.23 \\
sparse\_matrix & 24 & 3 & 120,000\(\to\)91,481\(\to\)76,297 & 12.7\% & 0.41/0.74 \\
sparse\_matrix & 48 & 3 & 240,000\(\to\)196,847\(\to\)165,358 & 13.1\% & 4.79/3.87 \\
sparse\_matrix & 64 & 1 & 320,000\(\to\)270,471\(\to\)233,703 & 11.5\% & timeout/303.23 \\
\bottomrule
\end{tabular}
\caption{Sparsity envelope for $p=2$ and $n=10{,}000$. The degree column is
the target average graph degree, and ``seeds'' is the number of independent
instances averaged. ``Raw/Hom total'' gives, in seconds, raw Ripser time and
the complete \textsc{HomTrim} pipeline time, including both trimming stages
and Ripser on the reduced graph.}
\label{tab:sparsity-compact}
\end{table}

%% file: tables/high_p_trim_only.tex
\begin{table}[ht]
\centering
\scriptsize
\setlength{\tabcolsep}{3.2pt}
\renewcommand{\arraystretch}{0.95}
\begin{tabular}{llrrrrrr}
\toprule
Family & \(n\) & parameter & \(p\) & \(m\) & \(m_H\) & red. & trim \\
\midrule
block\_clique & 5,000 & \(b=10\) & 5--7 & 22,500 & 12,451 & 44.7\% & 0.12 \\
block\_clique & 5,000 & \(b=12\) & 5--7 & 27,484 & 15,522 & 43.5\% & 0.29 \\
block\_clique & 10,000 & \(b=10\) & 5--7 & 45,000 & 24,891 & 44.7\% & 0.23 \\
sparse\_matrix & 5,000 & deg. 16 & 5--7 & 40,000 & 25,230 & 36.9\% & 0.15 \\
\bottomrule
\end{tabular}
\caption{High-degree trim-only stress tests. Here $m$ and $m_H$ are the edge
counts before and after \textsc{HomTrim}; ``red.'' is the percentage removed,
and trim time is in seconds. These rows are not PH-timed because Ripser
encounters a simplex-index overflow on the original instances.}
\label{tab:high-p-trim-only}
\end{table}

%% file: sections/conclusion.tex
\section{Concluding remarks}
\label{sec:conclusion}

The effect of deleting a generator on persistent homology is recorded exactly
by the future part of one filtered link. The resulting certificate is weaker
than conicality, it degrades gradually instead of failing outright, and for
deletions that share no simplex the costs combine by maximum rather than by
sum. Several questions remain open.

\paragraph{Dependent deletions.}
When deletions share simplices, we know no bound better than the sum of their
costs, and in the prototype runs of \Cref{sec:batch} short dependent
sequences nearly attained it. Whether long sequences can is open. A natural tool is the
filtration of $C_*(K,K\setminus S)$ by the number of elements of $S$ that a
simplex contains. It has as many layers as the largest number of elements of
$S$ in a common simplex, and its first layer splits over the elements of $S$
by the argument of \Cref{lem:splitting}. Identifying the homology of the
higher layers is the obstacle.

\paragraph{Order.}
The reduction depends on the processing order, and costs are not monotone
under deletion. We use reverse filtration order as a heuristic. It would be
useful to know a natural class of inputs on which some order is optimal, or
within a constant factor of optimal, for this certificate.

\paragraph{Implementation.}
Against domination-only preprocessing, the local homology is not yet
recovered end to end. The round scheme of \Cref{sec:batch} can be parallelized
within each round, and measuring it inside the full pipeline is the most
direct test of whether the backend savings can be kept.

\paragraph{Other generators and descriptors.}
The deletion theorem holds for generators of every dimension, but our
algorithms treat only vertices and edges. For the regularity descriptors, the
open question is whether slack diagrams carry useful information on real data;
our experiment only shows that they separate synthetic families built to
differ.

%% file: sections/backmatter.tex
\paragraph{Software and reproducibility.}
A package containing the implementation, experiment drivers, fixed seeds,
reference outputs, and table-generation instructions is available at
\url{https://github.com/siddharthpritam/homtrim}.

\paragraph{Declaration on the use of generative AI.}
The author used ChatGPT (GPT-5.6) and Claude (Sonnet 5 and Opus 5) for
assistance with code, experiments, editorial revision, consistency checks, and
exploratory proof discussions. The author reviewed and validated all
mathematical statements, proofs, code, experiments, and conclusions.

%% file: sections/appendix_missing_proofs.tex
\section{Deferred proofs}
\label{app:missing-proofs}

This appendix contains the two arguments deferred from the main text: the full
edge-filtration local-to-global proof and the stability proofs for regularity
slack.

\subsection{The edge local-to-global criterion}
\label{app:flag-local-global-proof}

\begin{proof}[Proof of \Cref{thm:flag-local-global}]
Assume first that $x_e<s$. Let $w'$ be obtained from $w$ by changing only the
value of $e$ from $x_e$ to $s$. The filtrations agree for $t<x_e$ and for
$t\geq s$. For $x_e\leq t<s$, we have $K_t^{w'}\subseteq K_t^w$. Set
$A_t=A_e^w(t)$. The simplices in $K_t^w$ but not in $K_t^{w'}$ are exactly
those containing $e$, and they form the block $e*A_t$. Every proper face of
$e$ joined with $A_t$ is already present, so
$$
K_t^w=K_t^{w'}\cup(e*A_t),
\qquad
K_t^{w'}\cap(e*A_t)=(\partial e)*A_t.
$$
Excision and \Cref{lem:join-pair} give
$$
H_k(K_t^w,K_t^{w'})
\cong
H_k\bigl(e*A_t,(\partial e)*A_t\bigr)
\cong\redH_{k-2}(A_t).
$$
The hypothesis makes these groups zero. The long exact sequence of the pair
therefore shows that $K_t^{w'}\hookrightarrow K_t^w$ induces an isomorphism
in every degree. These inclusions commute with the filtration maps, while
outside $[x_e,s)$ the two levels are identical. The levelwise isomorphisms
thus form an isomorphism
$$
H_k(\mathcal F_{w'})\cong H_k(\mathcal F_w)
$$
of persistence modules. If $s<x_e$, the same argument applies with the two
filtrations reversed.
\end{proof}

\subsection{Stability of obstruction and slack}
\label{app:stability-proofs}

\begin{proof}[Proof of \Cref{thm:universal-obstruction-stability}]
Consider first lower-star filtrations. If
$\tau\in A_\gamma^f(t)$, then every vertex $u\in\tau$ satisfies
$f(u)\leq t$. The bound $\|f-g\|_\infty\leq\eps$ implies
$g(u)\leq t+\eps$, and hence
$$
A_\gamma^f(t)\subseteq A_\gamma^g(t+\eps).
$$
Interchanging $f$ and $g$ gives the reverse shifted inclusion. These
inclusions commute with the filtration maps, so reduced homology gives an
$\eps$-interleaving in every degree. For edge-filtered flag complexes, the
same argument is applied to every edge needed to form a local simplex with
$\gamma$. Algebraic stability then gives the stated bottleneck bound
\cite{ChazalEtAl2016,Lesnick2015}.
\end{proof}

\begin{proof}[Proof of \Cref{prop:no-go}]
We give an explicit edge-filtration example. Fix numbers $1<s<T$ and let
$0<\eps<s-1$. Take an edge $e=uv$ with value $0$. Its common neighbours are
$$
r,a,b,c,d,h,j.
$$
Give all edges from $u$ or $v$ to these vertices value $-2$. Among the common
neighbours, include the six edges
$$
ra,\ ab,\ bc,\ rd,\ dh,\ hj
$$
with value $-1$. They form two paths sharing the vertex $r$. Add the edges
$cr$, $rb$, and $jr$, and include no other edges among the common neighbours.

Define two edge functions $f$ and $g$ that agree except on $cr$ and $rb$:
$$
f(cr)=f(rb)=s,
\qquad
g(cr)=s-\eps,
\qquad
g(rb)=s+\eps,
$$
and set $f(jr)=g(jr)=T$. Thus $\|f-g\|_\infty=\eps$.

At time $-1$, the filtered link of $e$ is a tree and is therefore acyclic. In
the $f$-filtration, the edge $cr$ closes the first four-cycle at the same time
that the diagonal $rb$ appears. The flag complex immediately contains the two
triangles that fill this cycle, so no obstruction is created. At time $T$, the
edge $jr$ closes the second four-cycle, which has no diagonal and gives a
persistent class in $\redH_1$.

In the $g$-filtration, the first cycle instead exists on
$[s-\eps,s+\eps)$: $cr$ creates it and $rb$ later fills it. The second cycle
still appears at $T$. Consequently, the regularity interval of $e$ containing
its anchor is
$$
[-1,T)
\quad\text{for }f,
\qquad
[-1,s-\eps)
\quad\text{for }g.
$$
The right endpoint and forward slack therefore change by
$T-s+\eps$, while the perturbation size is $\eps$. Letting $\eps$ tend to
zero rules out a universal Lipschitz constant.
\end{proof}

\begin{proof}[Proof of \Cref{thm:badset-slack-stability}]
Fix a generator and omit it from the notation. For a closed set
$C\subseteq\mathbb R$, write $\mu(C)$ for the distance from $0$ to $C$. The
bad set is never empty, since the filtered link is empty for small $t$ and
then $\redH_{-1}\neq0$; so $\mu(C)$ is finite and backward slacks are finite.
Forward slacks may be infinite, and in $\ell_\infty$-costs we use
$|\infty-\infty|=0$. If the generator is critical in $f$, then
$0\in C^f$ and $\mu(C^f)=0$. If it is regular, the open interval
$(-L^f,R^f)$ misses $C^f$ and its finite endpoints lie in $C^f$, so
$\mu(C^f)=m^f$, the $\ell_\infty$-distance from its RSD point to the boundary.
Because $d_H(C^f,C^g)\leq\delta$, every point of either set lies within
$\delta$ of a point of the other, and in particular
$|\mu(C^f)-\mu(C^g)|\leq\delta$.

Suppose first that $\mu(C^f)>\delta$. Then $\mu(C^g)>0$, so the generator is
regular in both filtrations. We show that $|R^f-R^g|\leq\delta$. If
$R^f<\infty$, then $R^f\in C^f$ and $R^f\geq\mu(C^f)>\delta$. A point
$c'\in C^g$ with $|c'-R^f|\leq\delta$ is therefore positive, and
$R^g\leq c'\leq R^f+\delta$. If $R^g<\infty$, choose $c\in C^f$ with
$|c-R^g|\leq\delta$. Then $c\geq-\delta$ and $|c|\geq\mu(C^f)>\delta$, so $c$
is positive and $R^f\leq c\leq R^g+\delta$. These two inequalities also show
that $R^f$ and $R^g$ are finite together. The same argument applies to the
backward slacks. The two RSD points are therefore within $\ell_\infty$-distance
$\delta$, and we match them. The case $\mu(C^g)>\delta$ is symmetric.

Otherwise $\mu(C^f)\leq\delta$ and $\mu(C^g)\leq\delta$. Every RSD point the
generator has, in either filtration, lies within $\delta$ of the boundary, and
we match it there. A generator that is critical in one filtration falls into
this case, since its bad set there contains $0$.

Doing this for every generator matches each point of one diagram either to a
point of the other or to the boundary, at cost at most $\delta$. Hence
$d_B^\partial\bigl(\RSD(\mathcal F_f),\RSD(\mathcal F_g)\bigr)\leq\delta$.
\end{proof}

\begin{proof}[Proof of \Cref{thm:robust-badset-stability}]
By \Cref{thm:universal-obstruction-stability}, the two obstruction barcodes
admit an $\eps$-matching in each degree. Take the union of these graded
matchings, and let $[b,d)$ be an $f$-bar of length greater than
$\tau+2\eps$. It cannot be matched to the diagonal, because its distance from
the diagonal is greater than $\eps$. It is therefore matched to a $g$-bar
$[b',d')$ satisfying
$$
|b-b'|\leq\eps,
\qquad
|d-d'|\leq\eps.
$$
The matched bar has length greater than $\tau$, and the support of $[b,d)$ is
contained in the $\eps$-neighbourhood of the support of $[b',d')$. Taking the
union over all such bars gives
$$
B_\gamma^{f,\tau+2\eps}
\subseteq N_\eps(B_\gamma^{g,\tau}).
$$
Interchanging $f$ and $g$ gives the second inclusion. Finally, centring a bad
set subtracts the anchor value. Since the two anchor values differ by at most
$\eps$, this adds at most another $\eps$ to the neighbourhood radius.
\end{proof}

%% file: sections/appendix_endpoint_blockers.tex
\section{Homological degrees at regularity endpoints}
\label{app:endpoint-bigrading}

The two slacks record when regularity ends. The graded obstruction barcode also
records which homological degree causes each endpoint.

Fix a regular generator $\gamma$ with finite regularity interval endpoints
$\ell_\gamma$ and $r_\gamma$. Treating bars with multiplicity, define the
left and right blocker multisets by
$$
\Bc_\gamma^-
=\{(q,[b,d)):[b,d)\in\Dgm_q(\Oc_\gamma),\ d=\ell_\gamma\},
$$
and
$$
\Bc_\gamma^+
=\{(q,[b,d)):[b,d)\in\Dgm_q(\Oc_\gamma),\ b=r_\gamma\}.
$$
Thus a left blocker is an obstruction that disappears where the regular window
begins, while a right blocker is one that appears where the window ends.

\begin{definition}[Endpoint-simple regularity interval]
A regularity interval is \emph{endpoint-simple} if each blocker multiset
contains exactly one bar. If those bars occur in degrees $q_\gamma^-$ and
$q_\gamma^+$, respectively, we call
$$
(q_\gamma^-,q_\gamma^+)
$$
the \emph{endpoint bi-grade} of $\gamma$.
\end{definition}

The bi-grade distinguishes regularity windows with the same two slacks but
different local mechanisms. For example, a right endpoint may be caused by a
link becoming disconnected in degree zero or by a cycle appearing in degree
one. If several bars begin or end together, the full blocker multiset is the
appropriate descriptor; no single bi-grade is assigned. Infinite endpoints
have no blocker on the corresponding side.

%% file: sections/appendix_rsd_experiment.tex
\section{A controlled experiment with slack summaries}
\label{app:rsd-experiment}

The main experiments use future obstruction as a deletion certificate. Here we
give a small synthetic illustration of its descriptive side. The question is
only whether summaries of the regularity slack diagram distinguish deliberately
different local obstruction mechanisms.

\paragraph{Setup.}
We generated three families of edge-filtered graphs. The \emph{clique-rich}
family has dense weighted blocks and sparse cross-block noise. The
\emph{cycle-rich} family contains local gadgets whose links begin as trees and
later acquire cycle-closing edges. The \emph{sparse-random} family is an
Erd\H{o}s--R\'enyi baseline with independent edge weights.

For each edge we used the degree bound $p=2$: its local link is regular when it
is nonempty, connected, and has trivial first homology over $\mathbb F_2$. We
then recorded five graph-level summaries: the fraction of regular edges, the
median and $90$th percentile of the margin, the fraction with finite forward
slack, and the fraction whose right endpoint is blocked in degree one.
\Cref{tab:rsd-structural-summaries} gives the mean and standard deviation over
$20$ graphs from each family.

\input{tables/rsd_structural_summaries}

\paragraph{Observed separation.}
The median margin alone does not separate the clique-rich and sparse-random
families: both medians are $0.148$, and classification by the median alone is
close to chance ($22/60$). Their regular-edge fractions and endpoint types are
different, however. The cycle-rich family has fewer regular edges, more finite forward
slacks, and more right endpoints caused by degree-one obstructions. In this
controlled setting, direction and obstruction degree therefore retain
information lost by reducing every edge to a single margin.

As a simple check, we performed leave-one-out nearest-centroid classification
on the graph-level summaries. For each of the $60$ graphs, the centroids were
computed from the other $59$ graphs and the held-out graph was assigned to the
nearest family. \Cref{tab:rsd-structural-classification} reports the results.

\input{tables/rsd_structural_classification}

The strong separation is expected because the families were constructed to
have different local mechanisms. It shows that the proposed summaries detect
those mechanisms; it does not establish predictive value on real data.

%% file: tables/rsd_structural_summaries.tex
\begin{table}[ht]
\centering
\scriptsize
\setlength{\tabcolsep}{2.8pt}
\renewcommand{\arraystretch}{1.05}
\begin{tabular}{@{}lrrrrrr@{}}
\toprule
Family
& \(|E|\)
& reg.
& med. \(m\)
& \(90\%\) \(m\)
& finite \(R\)
& \(H_1\)-block \(R\) \\
\midrule
clique-rich
& \(388.9\pm 6.1\)
& \(0.532\pm 0.023\)
& \(0.148\pm 0.025\)
& \(0.473\pm 0.052\)
& \(0.249\pm 0.037\)
& \(0.073\pm 0.021\) \\
cycle-rich
& \(280.8\pm 4.5\)
& \(0.144\pm 0.015\)
& \(0.118\pm 0.021\)
& \(0.224\pm 0.011\)
& \(0.644\pm 0.039\)
& \(0.292\pm 0.024\) \\
sparse-random
& \(449.0\pm 17.8\)
& \(0.206\pm 0.020\)
& \(0.148\pm 0.022\)
& \(0.454\pm 0.051\)
& \(0.368\pm 0.048\)
& \(0.000\pm 0.000\) \\
\bottomrule
\end{tabular}
\caption{RSD summaries for three synthetic families. The columns report edge
count, regular-edge fraction, median and $90$th-percentile margin, fraction
with finite forward slack, and fraction whose right endpoint is blocked in
degree one.}
\label{tab:rsd-structural-summaries}
\end{table}

%% file: tables/rsd_structural_classification.tex
\begin{table}[ht]
\centering
\scriptsize
\setlength{\tabcolsep}{5pt}
\renewcommand{\arraystretch}{1.05}
\begin{tabular}{@{}lr@{}}
\toprule
Features used & accuracy \\
\midrule
all five RSD summaries & \(60/60\) \\
regular fraction + right \(H_1\) blocker fraction & \(60/60\) \\
finite-forward fraction + right \(H_1\) blocker fraction & \(59/60\) \\
regular fraction only & \(58/60\) \\
right \(H_1\) blocker fraction only & \(59/60\) \\
median margin only & \(22/60\) \\
\bottomrule
\end{tabular}
\caption{Leave-one-out nearest-centroid check using graph-level RSD summaries.
These controlled results show separation among the three constructed families;
they are not an application-level validation.}
\label{tab:rsd-structural-classification}
\end{table}